\documentclass[12pt, reqno]{amsart}
\usepackage{a4wide}
\usepackage[utf8]{inputenc}
\usepackage{amsmath, amssymb}
\usepackage{mathrsfs}
\usepackage{microtype}
\usepackage{graphicx,wrapfig,lipsum}
\usepackage{epsfig,color}
\usepackage{latexsym}
\usepackage{pstricks}
\usepackage{amsthm}
\usepackage{etoolbox}
\usepackage{subcaption}
\usepackage{resizegather}
\usepackage{thmtools,thm-restate}
\usepackage{lipsum,afterpage}
\usepackage[numbers]{natbib}
\usepackage[toc,page]{appendix}

\usepackage{seqsplit}
\usepackage[justification=centering]{caption}
\usepackage{imakeidx}

\newtheorem{theorem}{Theorem}[section]

\newtheorem{lemma}[theorem]{Lemma}
\newtheorem{proposition}[theorem]{Proposition}
\newtheorem{definition}[theorem]{Definition}

\newtheorem{lemma and definition}[theorem]{Lemma and Definition}
\newtheorem*{theorem*}{Theorem}
\newtheorem*{conjecture*}{Conjecture}

\newtheoremstyle{myrem}%name
 {7pt}%Space above
 {7pt}%Space below
 {}%Body font
 { }%Indent amount
 {\bf}% Theorem head font
 {.}%Punctuation after theorem head
 { }%Space after theorem head 2
 {}%Theorem head spec (can be left empty, meaning ‘normal’)

 \theoremstyle{myrem}
 \newtheorem{remark}[theorem]{Remark}
 
 \newtheoremstyle{myrem}%name
 {7pt}%Space above
 {7pt}%Space below
 {}%Body font
 { }%Indent amount
 {\bf}% Theorem head font
 {.}%Punctuation after theorem head
 { }%Space after theorem head 2
 {}%Theorem head spec (can be left empty, meaning ‘normal’)

 \theoremstyle{myrem}
 \newtheorem*{example}{Example}
 
 \newtheoremstyle{myrem}%name
 {7pt}%Space above
 {5pt}%Space below
 {}%Body font
 { }%Indent amount
 {\bfseries}% Theorem head font
 {.}%Punctuation after theorem head
 { }%Space after theorem head 2
 {}%Theorem head spec (can be left empty, meaning ‘normal’)

 \theoremstyle{myrem}
 \newtheorem*{claim}{Claim}

\newcommand{\K}{\mathscr{K}}

\newcommand{\Pf}{\text{Pf}}

\newcommand{\Z}{\mathbb{Z}}

\newcommand{\R}{\mathbb{R}}
\newcommand{\trong}{t_{\text{in}}}
\newcommand{\ngoai}{t_{\text{out}}}

\allowdisplaybreaks

\title{A Pfaffian formula for the Ising partition function of surface graphs}
\author{Anh Minh Pham}

\address{Hanoi National University of Education, 136 Xuan Thuy Road, Cau Giay District, Hanoi, Vietnam}
\email{minhpa@hnue.edu.vn}

\begin{document}
	
	\begin{abstract}
		We give a Pfaffian formula to compute the partition function of the Ising model on any graph $G$ embedded in a closed, possibly non-orientable surface. This formula, which is suitable for computational purposes, is based on the relation between the Ising model on $G$ and the dimer model on its terminal graph $G^T$. By combining the ideas of Loebl-Masbaum \cite{Loeb11}, Tesler \cite{Tes2000}, Cimasoni \cite{Cim09, Cim10} and Chelkak-Cimasoni-Kassel \cite{Chel15}, we give an elementary proof for the formula.
	\end{abstract}
	\maketitle
\section{Introduction} The Ising model is probably one of the most famous models in statistical physics. It was introduced by Wilhelm Lenz in 1920 and named after his PhD student Ernst Ising. The model can be defined as follows. Given a finite graph $G$ with vertex set $V(G)$ and edge set $E(G)$, a \emph{spin configuration} on $G$ is a function $\sigma:V(G)\to \{\pm1\}$. If $G$ is endowed with a positive edge-weight system $J=(J_e)_{e\in E(G)}$, the \emph{energy} of such a configuration $\sigma$ is defined by $$\mathscr{H}(\sigma)=-\sum_{e=(u,v)\in E(G)}J_e\sigma(u)\sigma(v).$$

By fixing an \emph{inverse temperature} $\beta\geq0$, one can define a probability measure on the set $\Omega(G)$ of all spin configurations on $G$ by $$\mu^J_{\beta}(\sigma)=\frac{\exp(-\beta\mathscr{H}(\sigma))}{Z^J_\beta(G)},$$ where 
 the normalisation constant
$$Z^J_{\beta}(G):=\sum_{\sigma\in \Omega(G)} \exp(-\beta H(\sigma))$$ is called the \emph{partition function} of the \emph{Ising model} on $G$, or simply the \emph{Ising partition function} of $G$. By the observation of van der Waerden \cite{Wae41}, calculating the partition function $Z^J_{\beta}(G)$ of the Ising model on the graph $G$ with fixed $J$ and $\beta$ can be transformed to calculating its high-temperature expansion $$Z_{\mathcal{I}}(G)=\sum_{G'\subset G}\prod_{e\in E(G')}x_e,$$ where the sum is over all even subgraphs $G'$ of $G$, i.e., subgraphs such that each vertex of $G$ is met by an even number of edges of $G'$, and $x_e=\tanh(\beta J_e)$. The aim of this article is to give a formula for computing this high-temperature expansion efficiently, when the graph $G$ is embedded in a closed, possibly non-orientable surface.
\smallskip

Before stating the formula, let us first quickly summarise some results that have been achieved for this computation so far. The very first solution to the computation of the Ising partition function for a square lattice $G$ is due to Onsager \cite{Onsager44}. However his solution seemed to be very difficult, motivating Kac and Ward to find an easier approach \cite{Kac52}: they defined a so-called \emph{Kac-Ward matrix} $M(G)$ with rows and columns indexed by oriented edges of $G$, and proved that $Z_{\mathcal{I}}(G)=(\det M(G))^{1/2}$. This formula is now known as \emph{Kac-Ward formula}. Unfortunately, some arguments in \cite{Kac52} are not correct. It was only in 1999, by Dolbilin \emph{et al.} \cite{Dol99}, that  the first direct combinatorial proof of the Kac-Ward formula for any planar graph $G$ with straight edges was obtained. In the meantime, there were several attempts to prove the Kac-Ward formula by Sherman \cite{Sherman60}, Hurst-Green \cite{Hurst60}, Kasteleyn \cite{Kas61} and Fisher \cite{Fis61}  indirectly (still for a square lattice) by relating the high-temperature expansion with the dimer partition function of some associated graph. The latter then can be given by the Pfaffian of a skew-symmetric adjacency matrix of the associated graph. It should be recalled that this $``$Pfaffian method'' was then extended by Kasteleyn \cite{Kas63,Kas67} to any planar graph, by Galluccio-Loebl \cite{Gal99} to any graph embedded in an orientable surface with a particular drawing in the plane, by Tesler \cite{Tes2000} to any graph drawn in the plane, and by Cimasoni-Reshetikhin \cite{CimRes07,Cim09} to any graph embedded in a surface, to get a so-called \emph{Pfaffian formula} for the dimer partition function. Afterwards, inspired by the two approaches mentioned above, Cimasoni  \cite{Cim10} gave two different proofs for a general Kac-Ward formula, which holds for any graph embedded in an orientable surface. 

Let us discuss these two proofs in more detail. The first one, which is completely combinatorial, relies on the proof of Dolbilin \emph{et al.} \cite{Dol99} for planar graphs. In fact these authors developed nearly all of necessary tools to obtain the general Kac-Ward formula, but only a right notion of Kac-Ward matrices for surface graphs was missing. Actually it turns out that Kac-Ward matrices in this general case can be encoded using spin structures \cite[Definition 1]{Cim10}, and the proof in \cite{Dol99} can be adapted. The second proof of the general Kac-Ward formula is based on the \emph{Fisher correspondence} $G\mapsto \Gamma_G$ where $\Gamma_G$ is an associated graph still embedded in the same surface as $G$. Then by relating the determinant of generalised Kac-Ward matrices of $G$ with the Pfaffian of adjacency matrices of $\Gamma_G$, and using the Pfaffian formula to compute the dimer partition function of $\Gamma_G$ one obtains the general Kac-Ward formula.
\smallskip

Following these two proofs, it is possible to extend the general Kac-Ward formula to graphs embedded in (possibly) non-orientable surfaces. 
Let us consider the approach that goes through the dimer partition function of an associated graph. It was first observed by Kasteleyn \cite{Kas63} that one can transform the Ising model on a graph $G$ to the dimer model on its associated \emph{terminal graph} $G^T$. (However $G^T$ in general is not embedded in the same surface as $G$, so we can not apply the Pfaffian formula for $G^T$.) Building on this idea of Kasteleyn, Chelkak-Cimasoni-Kassel \cite{Chel15} also obtained the general Kac-Ward formula for graphs embedded in an orientable surface as in \cite{Cim10}. In more detail, they multiplied a generalised Kac-Ward matrix $KW_\lambda(G)$ encoded by a spin structure $\lambda$ with suitable ones to obtain a new matrix $\widehat{K}_\lambda$ which is skew-symmetric, whose determinant   is the same as that of $KW_\lambda(G)$. Since this new matrix can be thought of as a weighted adjacency matrix of the terminal graph $G^T$, its Pfaffian $\Pf\, (\widehat{K}_\lambda)$ counts weighted dimer configurations of $G^T$ with signs. The latter quantity then can be proved combinatorially to be equal to the Ising high-temperature expansion twisted by signs, and therefore we end up with the general Kac-Ward formula. We also would like to recall that, by \cite[Section 2.2]{Chel15}, the signs of dimer configurations of $G^T$ in the expansion of $\Pf\, (\widehat{K}_\lambda)$ amounts to saying that the  spin structure $\lambda$ is equivalent to a \emph{crossing orientation} of $G^T$ as defined by Tesler \cite{Tes2000}. Therefore it is possible to replace the generalised Kac-Ward matrix $KW_\lambda(G)$ by an adjacency matrix of $G^T$ with respect to the corresponding crossing orientation so that the general Kac-Ward formula for $G$ boils down to a Pfaffian formula for $G^T$.
\smallskip

 Inspired by the fact mentioned above, in this paper, we give a formula to express the Ising partition function of $G$ in terms of Pfaffians of suitable adjacency matrices of $G^T$. More precisely, we define \emph{good orientations} on $G^T$ with respect to a particular drawing of $G^T$ in the plane (cf. Definition \ref{def: Kasteleyn}). Then if the graph $G$ is embedded in a possibly non-orientable surface $\Sigma$, our formula reads
 $$Z_{\mathcal{I}}(G,x)=\frac{1}{2^{b_1/2}}\bigg|\sum_{q}\exp \left(\frac{i\pi}{4}\right)^{-\text{Br}(q)} \text{Pf}(A^{K_q,\omega}(G^T,x^T))\bigg|.$$ In this formula,  $b_1=\dim H_1(\Sigma;\Z_2)$, the sum is taken over all quadratic enhancements $q$ of $\Sigma$, $\text{Br}(q)$ denotes the Brown invariant of $q$, $K_q$ is an orientation derived from a good orientation $K$, and $A^{K_q,\omega}(G^T,x^T)$ is the twisted adjacency matrix of the terminal graph $(G^T,x^T)$ with respect to $K_q$ and a function $\omega$ characterising the orientability of $\Sigma$.
 This main result can be found in Theorem \ref{theo: general} with more details. A practical version can be found in Theorem \ref{theo: second}.
 
Last but not least, we would like to emphasise that one can use the whole proof of Tesler \cite{Tes2000} for crossing orientations on arbitrary graphs to prove our formula. However, it turns out that in the case of the terminal graph, the method developed by Tesler can be reduced extremely. Hence together with an idea of Loebl-Masbaum \cite{Loeb11}, we get an elementary, self-contained proof for the Pfaffian formula. Let us also mention that using the method of Dolbilin \emph{et al.} \cite{Dol99} (and then developed by Cimasoni \cite{Cim10}) one can have a Kac-Ward formula for graphs embedded in (possibly) non-orientable surfaces. In this case generalised Kac-Ward matrices can be encoded using pin$^-$ structures, which generalise spin structures to non-orientable surfaces \cite{KirTay90}. However our approach here does not need to use that geometrical notion.
\smallskip

The paper is organised as follows. In Section \ref{sec: Pf formula for Ising} we fix once and for all a drawing of the surface graph $G$  and define orientations of interest. We also state the Pfaffian formula in the orientable case (Theorems \ref{theo: first} and \ref{theo: main}). Section \ref{sec: Proof} is devoted to prove  Theorems \ref{theo: first} and \ref{theo: main}. We begin Section \ref{sec: Proof} with some preliminaries, recalling some results with detailed proof as a preparation for the non-orientable case. The main result of Section \ref{sec: Proof} is Proposition \ref{pro: main1}, which is still valid in the case of non-orientable surfaces. Using this result with slight modifications, in Section \ref{sec: non-orientable case Ising} we state and prove the general Pfaffian formula (Theorems \ref{theo: second} and \ref{theo: general}) for graphs embedded in non-orientable surfaces.
\subsection*{Acknowledgments} This work is supported by the NAFOSTED grant of Vietnam (Grant No. 101.04-2018.03). The author would like to thank his superadvisor David Cimasoni for helpful discussions.
\section{Statement of the Pfaffian formula in the orientable case}
\label{sec: Pf formula for Ising}
The most general Pfaffian formula for the Ising partition function of graphs embedded in surfaces can be found in Theorems \ref{theo: second} and \ref{theo: general} where surfaces are possibly non-orientable. However for simplicity, in this section let us first consider the case of orientable surfaces as a warm-up case, and study the case of non-orientable surfaces later.
\subsection{Graph drawing and good orientations}\label{subsec: graph drawing and good orientations}
Before stating the formula, let us first give some basic notions and terminology that are needed. More precisely, we shall describe a particular drawing of the graph $G$ and its terminal graph $G^T$ in the plane, together with some particular choice of orientations.

First of all, we will need a fixed drawing of $G$ in the plane which is specified as follows. Recall that $G$ is embedded in the orientable surface $\Sigma$ of genus $g$. Firstly, we represent $\Sigma$ as a planar $4g \text{-gon}$ $\mathcal{P}$ with $2g$ pairs of sides identified following the word $a_1b_1a_1^{-1}b_1^{-1}\cdots a_gb_ga_g^{-1}b_g^{-1}$, and draw $G$ in this polygon so that all the edges of $G$ intersect the $4g$ sides transversely. Secondly, for each pair of identified sides of $\mathcal{P}$, we add a strip to the outside of $\mathcal{P}$ connecting them (so that this strip is consistent with the side identification), and then extend all the edges of $G$ intersecting these two sides inside this strip (still following the side identification) so that their extended parts intersect each other transversely. Finally, ignoring all the sides of $\mathcal{P}$ as well as the added strips, we are left with the desired drawing of the graph $G$ (see Figure $\ref{fig: drawing}$). For further purposes, let us call edges of $G$ which are partly outside $\mathcal{P}$ the $\textit{outside}$ edges.
\begin{figure}[t]
	\centering
	\includegraphics[height=180pt]{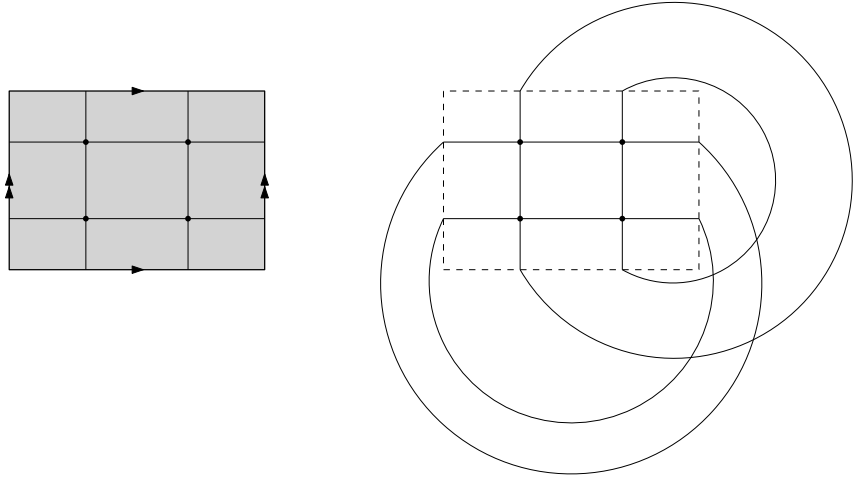}
	\caption{The $2\times 2$ square lattice embedded in the torus and its drawing in the plane.} \label{fig: drawing}
\end{figure}

Now let us continue with the terminal graph $G^T$, and particular orientations of interest. The idea of using the terminal graph to treat the Ising model is due to Kasteleyn \cite{Kas63} for a square lattice. Then Chelkak-Cimasoni-Kassel  \cite{Chel15} adapt this idea to graphs embedded in orientable surfaces. Here we also employ this idea to our setting. Recall that the $\emph{terminal graph}$ \index{terminal graph} $G^T$ (cf. \cite{Kas63}) associated to a given graph $G$  is the one obtained by replacing each vertex $v$ of degree $d(v)$ of $G$ by a complete graph $K_{d(v)}$ (see Figure $\ref{fig: terminal graph}$). Note that this transformation is local, so the fixed drawing of $G$ from the beginning also induces a drawing of $G^T$ in the plane. We shall say that an edge of $G^T$ is $\emph{short}$ if it is an edge of one of complete graphs, and $\emph{long}$ otherwise (i.e., if it comes from an edge of $G$). The long edges of $G^T$ which partly lie outside $\mathcal{P}$ will be also called $\emph{outside}$ edges, and $\emph{inside}$ otherwise. Given edge weights $x=(x_e)_{e \in E(G)}$ on $G$, we shall denote by $x^T$ the edge weights on $G^T$ obtained by assigning weight 1 to all long edges and weight $(x_e x_{e'})^{1/2}$ to the short edge corresponding to the two adjacent edges $e,e' \in E(G)$.

For further purposes, note that all the long edges of $G^T$ form a dimer configuration. Let us call the latter the \emph{standard dimer configuration} of $G^T$ and denote it by $D_0$.
\begin{figure}[h]
	\centering
	\includegraphics[height=100pt]{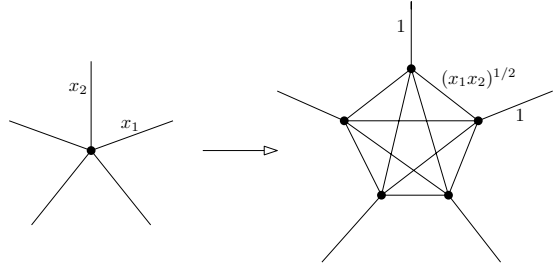}
	\caption{Local transformation at each vertex of $G$ to obtain $G^T$.}\label{fig: terminal graph}
\end{figure}

We now turn to some specific orientations on $G^T$. First of all, let us assume from now on that the plane is oriented counterclockwise. As mentioned above, we always fix the drawing of $G$ as well as the induced drawing of $G^T$ in the plane. With this drawing of $G^T$, we are interested in two types of faces: a face $f$ of $G^T$ is called $\emph{inside}$ if it is homeomorphic to a disc, and completely lies inside $\mathcal{P}$; the $\emph{outside}$ face $f_e$ corresponding to an outside edge $e$, by ignoring all the other outside edges, is the unique face formed by $e$ and some edges along the boundary of the subgraph obtained from $G^T$ by removing all the other outside edges (see Figure $\ref{fig: Kasteleyn}$). We also recall that each vertex $v\in V(G)$ of degree $d(v)$ gives rise to $d(v)$ vertices of the complete graph $K_{d(v)}$, and furthermore this complete graph can be drawn inside a small disc with the vertices on its boundary. Let us label these vertices by $1,\dots, d(v)$ with respect to the clockwise orientation (the starting vertex is not important). Inspired by Kasteleyn \cite{Kas63} and Tesler \cite{Tes2000} we define an orientation of interest as follows.
\begin{figure}[t]
	\centering
	\includegraphics[height=160pt]{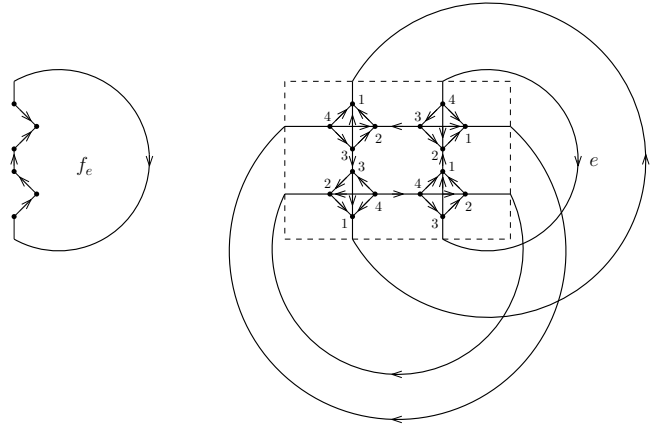}
	\caption{An outside face and a good orientation on the terminal graph of the $2\times 2$ square lattice.}\label{fig: Kasteleyn}
\end{figure}
\begin{definition} \label{def: Kasteleyn} An orientation $K$ on the edges of the terminal graph $G^T$ is called $\emph{good}$ \index{good orientation} if the three following conditions hold.
	\begin{enumerate} 
		\item[(i)] $K$ is from big-labelled vertices to small-labelled vertices on short edges.
		\item[(ii)]  For every inside face $f$ of $G^T$, the number $n^K(\partial f)$ of edges on its boundary $\partial f$ where the orientation of $\partial f$ is different from $K$ is odd.
		\item[(iii)] For every outside face $f_e$, $n^K(\partial f_e)$ is also odd.
	\end{enumerate}
\end{definition}
Note that there always exists good orientations (see Figure $\ref{fig: Kasteleyn}$ for example), however the proof of this fact is postponed until Proposition $\ref{pro: Kasteleyn}$. Let us pick a good orientation and denote it by $K_0$. We now show how to derive other specific orientations from $K_0$. Recall that the $4g\text{-gon}$ $ \mathcal{P}$ has $4g$ sides symbolised by the word $a_1b_1a_1^{-1}b_1^{-1}\cdots a_gb_ga_g^{-1}b_g^{-1}$, and this word induces the basis $\mathcal{B}:=\{[a_1],\dots,[a_g],[b_1],\dots,[b_g]\}$ of $H_1(\Sigma;\mathbb{Z}_2)$, the first homology group of $\Sigma$ with coefficients in $\Z_2$. For each $\Delta\in H_1(\Sigma;\mathbb{Z}_2)$, let us write $$\Delta=\sum_{i=1}^g\epsilon_i[a_i]+\sum_{i=1}^g\epsilon'_i[b_i],$$ with $\epsilon_i,\epsilon'_i\in \mathbb{Z}_2$. So by fixing the basis $\mathcal{B}$ (that we always do), one can identify $\Delta$ with $(\epsilon,\epsilon')\in \mathbb{Z}_2^{2g}$ where $\epsilon=(\epsilon_1,\dots ,\epsilon_g)$ and $\epsilon'=(\epsilon'_1,\dots ,\epsilon'_g)$ are elements of $\mathbb{Z}_2^g$. Then for each $\Delta\in H_1(\Sigma;\mathbb{Z}_2)$, or equivalently, for each $(\epsilon,\epsilon')\in \mathbb{Z}_2^{2g}$, let us denote by $K_{\epsilon,\epsilon'}$ the orientation obtained by inverting $K_0$ on every edge $e$ each time $e$ crosses a side $a_i$ (resp. $b_j$) of $\mathcal{P}$ with $\epsilon_i=1$ (resp. $\epsilon'_j=1$). The collection $\{K_{\epsilon,\epsilon'}:(\epsilon,\epsilon')\in \mathbb{Z}_2^{2g}\}$ consists of orientations of interest.
\smallskip

We are now ready to state the Pfaffian formula.
\subsection{Statement of the Pfaffian formula}\label{subsec: state Pf for Ising}
Recall that if $(G,x)$ is an edge-weighted graph of $2n$ vertices labelled by $\{1,\dots, 2n\}$ and $K$ is an arbitrary orientation on its edges, then the adjacency matrix \index{adjacency matrix} of $G$ with respect to $K$, denoted by $A^K(G,x)$, has entries defined by
\begin{equation*}
a_{ij}=\sum_{e=(i,j)}\epsilon^K_{ij}(e)x(e).
\end{equation*}
In this equality the sum is taken over all the edge $e$ of $G$ between two vertices $i,j$, and $$\epsilon^K_{ij}(e)= \left\{  \begin{array}{ll}
+1& \text{if $e$ is oriented by $K$ from $i$ to $j$};\\
-1& \text{if not}.
\end{array} \right. $$ The first version of the Pfaffian formula for the Ising partition function can be stated as follows.
\begin{theorem}
	\label{theo: first}
	Let $(G,x)$ be a weighted graph embedded in the orientable surface $\Sigma$ of genus $g$. Then the Ising partition function of $G$ is given by $$Z_{\mathcal{I}}(G,x)=\frac{1}{2^g}\bigg|\sum_{(\epsilon,\epsilon')\in \mathbb{Z}_2^{2g}}(-1)^{\sum\limits_{i=1}^g \epsilon_i\epsilon'_i} \text{Pf}(A^{K_{\epsilon,\epsilon'}}(G^T,x^T))\bigg|,$$ where $A^{K_{\epsilon,\epsilon'}}(G^T,x^T)$ is the adjacency matrix of the terminal graph $(G^T,x^T)$ with respect to the orientation $K_{\epsilon,\epsilon'}$.
\end{theorem}
Before going further, let us give an example to see how the formula above works.
\begin{figure}[h]
	\centering
	\includegraphics[height=120pt]{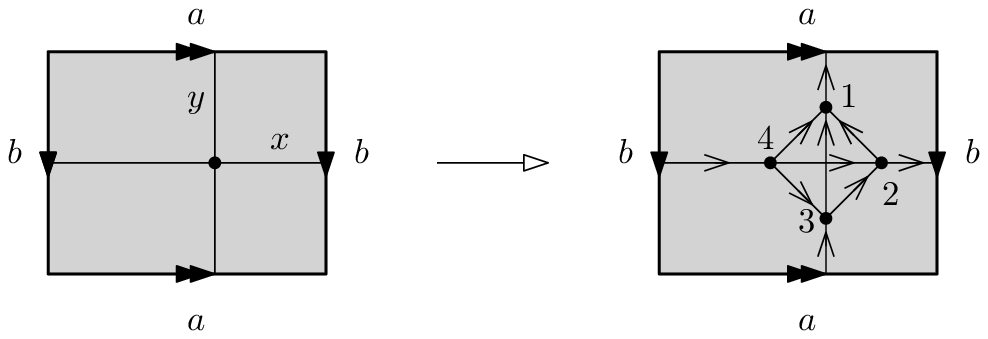}
	\caption{The graph $G\subset\mathscr{T}$ and its associated terminal graph $G^T$ with a good orientation $K$.\label{fig: exampleTorus}}
\end{figure}
\begin{example}
	Let $G$ be the $1\times 1$ square lattice embedded in the torus $\mathscr{T}$ with horizontal weight $x$ and vertical weight $y$. Label the vertices of $G^T$ as in Figure \ref{fig: exampleTorus}, and choose $K$ a good orientation. By definition we have $$A^{K}(G^T,x^T)=\begin{pmatrix} 0 & -\sqrt{xy} & 1-y & -\sqrt{xy}\\
	\sqrt{xy} & 0 & -\sqrt{xy} & 1-x \\
	-1+y & \sqrt{xy} & 0 & -\sqrt{xy} \\
	\sqrt{xy} & -1+x & \sqrt{xy} & 0
	\end{pmatrix}$$ whose Pfaffian is $\Pf(A^{K}(G^T,x^T))=\Pf(A^{K_{0,0}}(G^T,x^T))=xy-1+x+y$. Note that in this case $K_{10}$ is obtained by inverting $K$ on the long edge connecting the vertices labelled 1 and 3, and similarly for $K_{01}$ and $K_{11}$. Hence one can find easily $$\Pf (A^{K_{1,0}}(G^T,x^T))=xy+1-x+y,$$ $$\Pf (A^{K_{0,1}}(G^T,x^T))=xy+1+x-y,$$ 
	$$\Pf (A^{K_{1,1}}(G^T,x^T))=xy-1-x-y.$$
	Then the formula in Theorem \ref{theo: first} gives
	\begin{multline*}
	Z_{\mathcal{I}}(G,x)=\frac{1}{2}\big|\Pf (A^{K_{0,0}}(G^T,x^T))+\Pf (A^{K_{0,1}}(G^T,x^T))\\+\,\Pf (A^{K_{1,0}}(G^T,x^T))-\,\Pf (A^{K_{1,1}}(G^T,x^T)) \big|,\end{multline*} leading to $Z_{\mathcal{I}}(G,x)=xy+x+y+1$, which is trivially correct.
\end{example}
\smallskip
We now continue with the second version of the Pfaffian formula. To do so one needs the terminology of quadratic forms and Arf invariant that we now very briefly recall. Readers are referred to \cite[Chapter 9]{Saveliev12} for more details.

Given $(H,\cdot)$ a vector space of finite dimension over $\mathbb{Z}_2$ together with an alternating bilinear form $(\cdot)$, a function $q:H\rightarrow \mathbb{Z}_2$ is called 
a $\emph{quadratic form}$ \index{quadratic form} on $(H,\cdot)$ if for every $x,y\in H$ we have $$q(x+y)=q(x)+q(y)+x\cdot y.$$ Note that every quadratic form is completely determined by its value on a basis of $H$, and these values can be chosen freely. Moreover, if $(\cdot)$ is non-degenerate then one can define the $\emph{Arf invariant}$ of a quadratic form $q$ as a number $\text{Arf}(q) \in\mathbb{Z}_2$ by the formula 
$$(-1)^{\text{Arf}(q)}=\frac{1}{\sqrt{|H|}}\sum_{x\in H}(-1)^{q(x)}.$$

Coming back to our context, we are interested in quadratic forms on $(H_1(\Sigma;\Z_2),\cdot)$ where $(\cdot)$ is the intersection form on  $H_1(\Sigma;\Z_2)$.  Denote the set of all such forms by $\mathcal{Q}(\Sigma)$. Now imagine that the $4g \text{-gon}$ $\mathcal{P}$ representing $\Sigma$ is embedded in $\mathbb{R}^3$ together with the added strips. One can expand all the strips suitably so that we get back the original surface $\Sigma$ with $\mathcal{P}$ as its subset. Then one can see that $H_1(\Sigma;\mathbb{Z}_2)$ is isomorphic to $ H_1(\Sigma,\mathcal{P};\mathbb{Z}_2)$, the first relative homology group of $\Sigma$ with respect to $\mathcal{P}$. Thus a quadratic form $q\in\mathcal{Q}(\Sigma)$ can be considered as a $\mathbb{Z}_2-$valued function on $H_1(\Sigma,\mathcal{P};\mathbb{Z}_2)$. Let $q_0$ be the quadratic form getting value 0 on the basis $\mathcal{B}$. For each quadratic form $q\in \mathcal{Q}(\Sigma)$, we define $K_q$ as the orientation obtained by inverting $K_0$ on every outside edge $e$ such that $q([e])\not= q_0([e])$. (It is trivial that by this notation one has $K_0=K_{q_0}$.) Then we can state the second version of the Pfaffian formula as follows.
\begin{theorem}
	\label{theo: main}
	Let $(G,x)$ be a weighted graph embedded in an orientable surface $\Sigma$ of genus $g$. Then the Ising partition function of $G$ is given by $$Z_{\mathcal{I}}(G,x)=\frac{\epsilon_0}{2^g}\sum_{q\in \mathcal{Q}(\Sigma)}(-1)^{\text{Arf}(q)} \text{Pf}(A^{K_q}(G^T,x^T)),$$ where $\epsilon_0=\pm 1$ is a constant. In this formula, $A^{K_q}(G^T,x^T)$ is the adjacency matrix of the terminal graph $(G^T,x^T)$ with respect to the orientation $K_q$.
\end{theorem}
\section{Proof of the Pfaffian formula}\label{sec: Proof}
\subsection{Preliminaries} \label{subsec: preliminaries} In this subsection, we recall some results of Loebl-Masbaum \cite{Loeb11}, Tesler \cite{Tes2000} and Chelkak-Cimasoni-Kassel \cite{Chel15} as a preparation for the proofs of Theorems \ref{theo: first} and \ref{theo: main}; the easy proofs are included for completeness.

Let us first come back to the drawing of $G$ and see how the intersection form $(\cdot)$ on $H_1(\Sigma;\mathbb{Z}_2)$ relates to the one in the plane. Recall that in the previous section, we identify $H_1(\Sigma;\mathbb{Z}_2)$ with $H_1(\Sigma,\mathcal{P};\mathbb{Z}_2)$.
\begin{lemma}\label{lem: intersection} Let $e_1,e_2$ be two outside edges of $G$ which induce two relative homology classes $[e_1],[e_2]\in H_1(\Sigma,\mathcal{P};\mathbb{Z}_2)$. Then we have the following equality modulo 2 $$[e_1]\cdot[e_2]=e_1\cdot e_2,$$ where the intersection in the left-hand side is in $H_1(\Sigma;\mathbb{Z}_2)$, while the right-hand side is the geometric intersection number in the plane.
\end{lemma}
\begin{proof}
	For $i=1,2$ let $e_i'$ be any loop obtained from $e_i$ by connecting their two endpoints in an arbitrary way. By the fixed drawing of $G$, it is obvious that $$[e_1]\cdot[e_2]+e_1\cdot e_2=e_1'\cdot e_2'.$$ The last quantity is always 0 modulo 2 because any two loops in the plane, if cross each other transversely, will cross an even number of times.
\end{proof}
Consequently, we get the following result, which is due to Loebl-Masbaum  \cite[Proposition 2.6]{Loeb11}.
\begin{proposition} \label{pro: outside intersections}
	For every $M\subset E(G)$, the number $t(M)$ of its self-intersections in the plane has the same parity as $q_0([M])$, where $[M]\in H_1(\Sigma, \mathcal{P};\mathbb{Z}_2)$ is the relative homology class of $M$.\end{proposition}
\begin{proof}
	Assume that $M=\{e_k\}_{1\le k\le m}$, and for each $k$ write $[e_k]=\sum_{i=1}^g \alpha_i^k [a_i]+\sum_{i=1}^g \beta_i^k [b_i]$ with $\alpha_i^k,\beta_i^k\in \mathbb{Z}_2$. By the fixed drawing of $G$ and using the fact that $[a_i]\cdot[b_j]=\delta_{ij}$ and $[a_i]\cdot[a_j]=[b_i]\cdot[b_j]=0$, one sees that the number $t(e_k)$ of self-intersections of $e_k$ is equal to $\sum_{i=1}^g \alpha_i^k\beta_i^k$. Using this fact and the definition of quadratic forms, as well as Lemma \ref{lem: intersection}, we have the following equalities modulo 2
	\begin{align}
	q_0([M])&=q_0(\sum_{k}[e_k])=\sum_{k} q_0([e_k])+\sum_{k<l} [e_k]\cdot[e_l] \nonumber \\
	&= \sum_{k} q_0(\sum_{i=1}^g \alpha_i^k [a_i]+\sum_{i=1}^g \beta_i^k [b_i])+\sum_{k<l} [e_k]\cdot[e_l] \nonumber \\
	&=\sum_k \big(q_0(\sum_{i=1}^g \alpha_i^k [a_i])+q_0(\sum_{i=1}^g \beta_i^k [b_i])+(\sum_{i=1}^g \alpha_i^k[a_i])\cdot (\sum_{i=1}^g\beta_i^k[b_i])\big)+\sum_{k<l} [e_k]\cdot[e_l] \nonumber \\
	&=\sum_{k} (\sum_{i=1}^g \alpha_i^k[a_i])\cdot (\sum_{i=1}^g\beta_i^k[b_i])+\sum_{k<l} [e_k]\cdot[e_l] =\sum_{k}\sum_{i=1}^g\alpha_i^k\beta_i^k+\sum_{k<l} [e_k]\cdot[e_l] \nonumber \\
	&=\sum_{k} t(e_k)+\sum_{k<l} e_k\cdot e_l =t(M). \nonumber
	\end{align}Note that the fifth equality comes from the fact that $q_0([a_i])=q_0([b_i])=0$. The proof is concluded.
\end{proof}
We will also need the two following properties of the Arf invariant which can be found in \cite[Lemma 2.10]{Loeb11} and in \cite[Lemma 1]{CimRes07}. Alternatively, the reader can wait until Lemma \ref{lem: Brown} from which we can get back these two properties.
\begin{lemma}\label{lem: Arf} Let $q$ be a quadratic form on $(H,\cdot)$, then we have:
	\begin{enumerate}
		\item[(i)] The equality $\frac{1}{\sqrt{|H|}}\sum\limits_{q\in \mathcal{Q}(H,\cdot)}(-1)^{\text{Arf}(q)+q(x)}=1$ holds for every $x\in H$, where the sum is taken over all the set of quadratic forms.
		\item[(ii)] If $q'$ is also a quadratic form, then $\text{Arf}(q)+\text{Arf}(q')=q(\Delta)=q'(\Delta)$, where $\Delta\in H$ satisfies $(q+q')(x)=\Delta \cdot x$ for every $x\in H$.\qed
	\end{enumerate}
\end{lemma} 

Next let us move on by showing how to transform the partition function of the Ising model on $G$ (twisted with some signs) to that of the dimer model on $G^T$. To do so, we need some notions and terminology. Recall that $\mathcal{D}(G^T)$ denotes the set of dimer configurations on $G^T$. We also write $\mathcal{D}_{2n}$ for $\mathcal{D}({K_{2n})}$. For $D\in\mathcal{D}(G^T)$, we define $\trong(D)$ (resp. $\ngoai(D)$) its number of self-intersections lying inside (resp. outside) the polygon $\mathcal{P}$. Note that $\trong(D)$ counts the number of crossings created by short edges of $D$, while $\ngoai(D)$ counts the number of those created by long edges. It is trivial that the total number $t(D)$ of self-intersections of $D$ satisfies $t(D)=\trong(D)+\ngoai(D)$. Representing $K_{d(v)}$ inside a closed disc with its vertices on the boundary of the disc, we see that, for every $D\in \mathcal{D}(G^T)$, the parity of $\trong(D)$ does not depend on the way $K_{d(v)}$ is drawn. Furthermore, if $d(v)$ is even, we have the following result which is due to Chelkak-Cimasoni-Kassel \cite[Lemma 2.1]{Chel15}.
\begin{lemma} \label{lem: dimer} For any integer $n\geq 1$, $\sum_{D\in \mathcal{D}_{2n}}(-1)^{\trong(D)}=1$.
\end{lemma}
\begin{proof}
	Fix two vertices of $K_{2n}$ and consider the map $\sigma: \mathcal{D}_{2n}\rightarrow\mathcal{D}_{2n}$ given by exchanging them. The fixed point set $\text{Fix}(\sigma)$ consists of all the dimer configurations of $K_{2n}$ matching these two vertices. Since $t(D)=t(\sigma(D))+1$ for $D\notin\text{Fix}(\Sigma)$ and $\sigma$ is bijective, we get
	\begin{eqnarray*}
		\sum_{D\in \mathcal{D}_{2n}}(-1)^{\trong(D)} &=& \sum_{D\in \text{Fix}(\sigma)}(-1)^{\trong(D)}+\sum_{D\notin \text{Fix}(\sigma)}(-1)^{\trong(D)}\\
		&=&  \sum_{D'\in \mathcal{D}_{2(n-1)}}(-1)^{\trong(D')}+\dfrac{1}{2}\sum_{D\notin \text{Fix}(\sigma)}\left( (-1)^{\trong(D)}+(-1)^{\trong(\sigma(D))}\right )\\
		&=&  \sum_{D'\in \mathcal{D}_{2(n-1)}}(-1)^{\trong(D')}.
	\end{eqnarray*}
	The lemma now follows by induction on $n\geq 1$.
\end{proof}
Now let us denote by $Z_1(G;\mathbb{Z}_2)$ the set of all 1-cycles modulo 2 in $G$, that is, $$Z_1(G;\mathbb{Z}_2)=\{ P=\sum_e e\in C_1(G;\mathbb{Z}_2):\partial_1 P=0\in C_0(G;\mathbb{Z}_2)\}$$ where the sum is over some finite set of different edges of $G$. In other words, $Z_1(G;\mathbb{Z}_2)$ is exactly the set of even subgraphs of $G$. Setting $x(P):=\prod_{e\in P}x_e$ for each $P\in Z_1(G;\mathbb{Z}_2)$, we have the following result.
\begin{lemma} \label{lem: Z}
	Setting $Z^q(G,x):= \sum\limits_{\alpha\in H_1(\Sigma;\mathbb{Z}_2)}(-1)^{q(\alpha)}\sum\limits_{[P]=\alpha}x(P)$ for each quadratic form $q\in \mathcal{Q}(\Sigma)$, we have
	\begin{equation}
	Z^q(G,x)= \sum_{D\in \mathcal{D}(G^T)}(-1)^{q([D\Delta D_0])+\trong(D)}x^{T}(D).
	\end{equation}
	In this formula, $D_0$ is the standard dimer configuration of $G^T$, i.e., the one consisting of all the long edges.
\end{lemma}
Note that this equality was mentioned in \cite[Subsection 4.2]{Chel15} but not with full details, so let us give its proof here for the sake of completeness.
\begin{proof}[Proof of Lemma $\ref{lem: Z}$]
	Recall that $(G^T,x^T)$ is the associated terminal graph of the embedded graph $(G,x)$ in $\Sigma$. Given a dimer configuration $D\in \mathcal{D}(G^T)$, let $D_G$ denote the subgraph of $G$ given by the edges of $G$ corresponding to the long edges of $D$. We claim that $G\setminus D_G$ is an even subgraph of $G$. Indeed, consider a vertex $v$ of $G$, and denote by $m(v)$  the degree of $v$ in $D_G$. In and around the complete graph $K_{d(v)}$, the number of vertices that $D$ matches is exactly $d(v)+m(v)$, which is even since $D$ is a dimer configuration. Since the degree of $v$ in $G\setminus D_G$ is equal to $d(v)-m(v)$, the claim follows. Therefore the assignment $D\mapsto G\setminus D_G$ defines a map $\varphi: \mathcal{D}(G^T)\rightarrow Z_1(G;\mathbb{Z}_2)$. Moreover, we have $x(P)=x(\varphi(D))=x^T(D)$, and $\varphi^{-1}(P)=\prod_{v\in V}\mathcal{D}_{2n(v)}$, where $V$ denotes the set of all vertices of $G$, and $2n(v)$ denotes the degree of $v$ in $P$. From these facts and Lemma \ref{lem: dimer}, we can write
	\begin{eqnarray}
	Z^q(G,x)&=& \sum_{\alpha\in H_1(\Sigma;\mathbb{Z}_2)}(-1)^{q(\alpha)}\sum_{[P]=\alpha}x(P) \nonumber\\
	&=&\sum_{\alpha\in H_1(\Sigma;\mathbb{Z}_2)}(-1)^{q(\alpha)}\sum_{[P]=\alpha} \big( \prod_{v\in V}\sum_{D_v\in \mathcal{D}_{2n(v)}}(-1)^{\trong(D_v)} \big)x(P) \nonumber\\
	&=&\sum_{\alpha\in H_1(\Sigma;\mathbb{Z}_2)}(-1)^{q(\alpha)}\sum_{[P]=\alpha}\sum_{D\in \varphi^{-1}(P)}(-1)^{\trong(D)}x^{T}(D)\nonumber \\
	&=& \sum_{D\in \mathcal{D}(G^T)}(-1)^{q([G\setminus D_G])+\trong(D)}x^{T}(D)\nonumber\\
	&=&\sum_{D\in \mathcal{D}(G^T)}(-1)^{q([D\Delta D_0])+\trong(D)}x^{T}(D).\nonumber
	\end{eqnarray}
	The last equality comes from the fact that $G\setminus D_G$ is homologous to $D\Delta D_0$.
\end{proof}
Let us now discuss good orientations. The proposition below shows that they do exist, and how to construct one of them.
\begin{proposition} \label{pro: Kasteleyn}
	Fixing a drawing of the graph $G$ and its terminal graph $G^T$ as before, there always exists good orientations on $G^T$.
\end{proposition}
\begin{proof}
	Firstly let us label vertices in each complete graph increasingly with respect to the clockwise orientation, and orient short edges by condition $(i)$ of Definition $\ref{def: Kasteleyn}$. Secondly let us pick an arbitrary orientation $K$ on the inside long edges of $G^T$. For each inside face $f$ of $G^T$ with $n^K(\partial f)$ even, we draw a path from the interior of $f$ to the outside of $\mathcal{P}$ so that this path only crosses long edges transversely, and invert $K$ on each inside long edge crossed by this path. Repeating this procedure for each face $f$ with $n^K(\partial f)$ even, together with the orientation on short edges, we get an orientation satisfying both conditions $(i)$ and $(ii)$ of Definition $\ref{def: Kasteleyn}$. Finally, the orientations on the outside long edges are determined uniquely by condition $(iii)$ of Definition $\ref{def: Kasteleyn}$. 
\end{proof}
\begin{remark}\label{rem: good orientations}
	The definition of good orientations can be extended to subgraphs of $G^T$ in the following sense. Suppose that $G'$ is a subgraph of $G^T$ obtained by removing some interior short edges of complete graphs so that $G'$ has no self-intersections inside $\mathcal{P}$. A good orientation $K$ on $G^T$ restricts to $K'$ on $G'$, which is still a good orientation by Definition $\ref{def: Kasteleyn}$. Indeed, the conditions $(i)$ and $(iii)$ hold for $K'$ immediately. We only need to check that the condition $(ii)$ also holds for newly-created inside faces of $G'$. Let $f'$ be such a face. By the way we label vertices in each complete graph and by condition $(i)$, $n^{K'}(\partial f')$ is always equal to 1, which is odd. We will use this fact later in our proof.
\end{remark}
\subsection{Proof of Theorems $\ref{theo: first}$ and $\ref{theo: main}$}\label{subsec: proof}
The aim of this subsection is to prove the two theorems stated in Subsection \ref{subsec: state Pf for Ising}. In order to do that, we first recall several facts about Pfaffians of adjacency matrices.

By definition the \emph{Pfaffian}\index{Pfaffian} of a skew-symmetric matrix $A=(a_{ij})_{1\leq i,j\leq 2n}$ is given by $$\Pf(A)=\frac{1}{2^n n!}\sum_{\sigma\in S_{2n}} \text{sign} (\sigma)a_{\sigma(1)\sigma(2)}\cdots a_{\sigma(2n-1)\sigma(2n)}.$$
As $A$ is skew-symmetric, each term in the right-hand side corresponding to $\sigma$ only depends on the matching of $\{1,\dots ,2n\}$ into $n$ unordered pairs $\{\sigma(1),\sigma(2)\},\dots ,\{\sigma(2n-1),\sigma(2n)\}$. Since there are exactly $2^n n!$ permutations representing a same matching, one can write $$\Pf(A)=\sum_{[\sigma]}\text{sign}(\sigma)a_{\sigma(1)\sigma(2)}\cdots a_{\sigma(2n-1)\sigma(2n)},$$ where the sum is over the set of matchings of $\{1,\dots ,2n\}$.

When $A$ is $A^k(G,x)$, the adjacency matrix of an edge-weighted graph $(G,x)$ with respect to an orientation $K$, it is clear from the previous equation that a matching of the vertices of $G$ contributes to the Pfaffian of $A^k(G,x)$ if and only if it is realised by a dimer configuration of $G$. Therefore we can write \begin{equation} \label{eq: Pfaffian} \Pf(A^K(G))=\sum\limits
_{D\in \mathcal{D}(G)} \epsilon^K(D)x(D),
\end{equation}
where the sum is taken over the set $\mathcal{D}(G)$ of dimer configurations of $G$, $x(D):=\prod\limits_{e\in D}x(e)$ and the sign $\epsilon^K(D)=\pm 1$ can be computed as follows. If the dimer configuration $D$ is given by edges $e_l$ matching vertices $i_l$ and $j_l$ for $1\leq l\leq n$, let $\sigma$ denote the permutation mapping the set $\{1,2,\dots,2n-1, 2n\}$ to $\{i_1,j_1,\dots,i_n,j_n\}$. The sign then is equal to 
\begin{equation}\label{eq: sign}
\epsilon^K(D)=\text{sign}(\sigma)\prod_{l=1}^n\epsilon_{i_lj_l}^K(e_l),
\end{equation} recalling that $\epsilon^K_{i_lj_l}(e_l)= 
1$ if $e_l$ is oriented by $K$ from $i_l$ to $j_l$, and $\epsilon^K_{i_lj_l}(e_l)= -1$ otherwise.
\smallskip

The proofs of Theorems \ref{theo: first} and \ref{theo: main} rely on the following key result.
\begin{proposition} \label{pro: main1}
	For every dimer configuration $D$ of $G^T$, we have
	$$\epsilon^{K_0}(D)=\epsilon_0(-1)^{t(D)}$$
	where $\epsilon_0=\pm 1$ is a constant.
\end{proposition}
It should be mentioned that, by this proposition, the good orientation $K_0$ by our definition is a $\emph{crossing orientation}$ as defined by Tesler \cite{Tes2000}. Moreover, using this property of crossing orientations and some properties of the Arf invariant, Loebl and Masbaum \cite{Loeb11} gave a direct proof of the Pfaffian formula for the dimer partition function. It turns out that their proof can be adapted to our situation for the Ising partition function, as we will show now. The proof of Proposition $\ref{pro: main1}$ will be left until the end of this section.
\begin{proof}[Proof of Theorem \ref{theo: main}]
	First of all we recall that the orientation $K_q$ is obtained by inverting $K_0$ on every edge $e$ such that $q([e])\neq q_0([e])$. Hence by Equation ($\ref{eq: sign}$) we get $ \epsilon^{K_q}(D)=\epsilon^{K_0}(D)(-1)^{|\{e\in D: q([e])\neq q_0([e])\}|}$. By definition of quadratic forms, we can write $$ q([D])-q_0([D])=\sum_{e\in D}q([e])-\sum_{e\in D}q_0([e])= |\{e\in D: q([e])\neq q_0([e])\}| \in\Z_2,$$ which implies that $ \epsilon^{K_q}(D)=\epsilon^{K_0}(D)(-1)^{q([D])-q_0([D])}$. Now by Proposition \ref{pro: outside intersections} we have $q_0([D])=\ngoai(D)$ modulo 2 for every $D\in \mathcal{D}(G^T)$. This fact and Proposition $\ref{pro: main1}$ lead to
	\begin{eqnarray*}
		\epsilon^{K_q}(D)&=&\epsilon_0(-1)^{t(D)}(-1)^{q([D])-q_0([D])}=\epsilon_0(-1)^{\trong(D)+\ngoai(D)}(-1)^{q([D])-q_0([D])}\\
		&=& \epsilon_0(-1)^{\trong(D)+q_0([D])}(-1)^{q([D])-q_0([D])}=\epsilon_0(-1)^{\trong(D)+q([D])}.
	\end{eqnarray*}
	Secondly, for $D_0$ the standard dimer configuration of $G^T$, its relative homology class $[D_0]\in H_1(\Sigma,\mathcal{P};\mathbb{Z}_2)$ induces a dual, denoted by $[D_0]^*\in Hom(H_1(\Sigma,\mathcal{P};\mathbb{Z}_2);\mathbb{Z}_2)=H^1(\Sigma,\mathcal{P};\mathbb{Z}_2)$, which can be simply given by $[D_0]^*([D]):=[D]\cdot [D_0]$ for every $[D]\in H_1(\Sigma,\mathcal{P};\mathbb{Z}_2)$. Note that since $[D_0]^*$ is linear, $q+[D_0]^*$ is still a quadratic form. Using the equation above for  $q+[D_0]^*$ together with Lemma $\ref{lem: Z}$, one can write 
	\begin{eqnarray}
	Z^q(G,x)&=&\sum_{D\in \mathcal{D}(G^T)}(-1)^{q([D\Delta D_0])+\trong(D)}x^T(D) \nonumber \\ \nonumber
	&=&(-1)^{q([D_0])}\sum_{D\in \mathcal{D}(G^T)}(-1)^{q([D])+[D]\cdot [D_0]+\trong(D)}x^T(D)\\  \nonumber
	&=&(-1)^{q([D_0])}\sum_{D\in \mathcal{D}(G^T)}(-1)^{(q+D_0^*)([D])+\trong(D)}x^T(D)\\  \nonumber
	&=&\epsilon_0(-1)^{q([D_0])}\sum_{D\in \mathcal{D}(G^T)}\epsilon^{K_{q+[D_0]^*}}(D)x^T(D)\\ 
	&\overset{(\ref{eq: Pfaffian})}{=}& \epsilon_0(-1)^{q([D_0])} \text{Pf}(A^{K_{q+[D_0]^*}}(G^T,x^T)). \label{eq: Z^q}
	\end{eqnarray}
	Finally, let us recall that our purpose is to give a formula to calculate the Ising partition function on $G$ $$Z_{\mathcal{I}}(G,x)=\sum_{G'\subset G}\prod_{e\in E(G')}x_e,$$ where the sum is over all even subgraphs $G'$ of $G$, or in other words, over the set $Z_1(G;\mathbb{Z}_2)$ of all 1-cycles modulo 2 in $G$. Therefore the Ising partition function can be rewritten as
	\begin{eqnarray*}
		Z_{\mathcal{I}}(G,x)&=&\sum_{P\in Z_1(G;\mathbb{Z}_2)}x(P)=\sum_{\alpha\in H_1(\Sigma;\mathbb{Z}_2)}\sum_{[P]=\alpha} x(P)\\
		&\overset{\text{Lem.}\, \ref{lem: Arf}(i)}{=}&\sum_{\alpha\in H_1(\Sigma;\mathbb{Z}_2)}\bigg(\frac{1}{2^g}\sum_{q\in \mathcal{Q}(\Sigma)}(-1)^{\text{Arf}(q)+q(\alpha)}\bigg)\sum_{[P]=\alpha} x(P)\\
		&=&\frac{1}{2^g}\sum_{q\in \mathcal{Q}(\Sigma)}(-1)^{\text{Arf}(q)}\sum_{\alpha\in H_1(\Sigma;\mathbb{Z}_2)}(-1)^{q(\alpha)}\sum_{[P]=\alpha} x(P)\\
		&=&\frac{1}{2^g}\sum_{q\in \mathcal{Q}(\Sigma)}(-1)^{\text{Arf}(q)}Z^q(G,x)\\
		&\overset{(\ref{eq: Z^q})}{=}&\frac{\epsilon_0}{2^g}\sum_{q\in \mathcal{Q}(\Sigma)}(-1)^{\text{Arf}(q)+q([D_0])}\text{Pf}(A^{K_{q+[D_0]^*}}(G^T,x^T))\\
		&\overset{\text{Lem.}\, \ref{lem: Arf}(ii)}{=}&\frac{\epsilon_0}{2^g}\sum_{q\in \mathcal{Q}(\Sigma)}(-1)^{\text{Arf}(q+[D_0]^*)}\text{Pf}(A^{K_{q+[D_0]^*}}(G^T,x^T))\\
		&=&\frac{\epsilon_0}{2^g}\sum_{q\in \mathcal{Q}(\Sigma)}(-1)^{\text{Arf}(q)}\text{Pf}(A^{K_{q}}(G^T,x^T)).
	\end{eqnarray*}This concludes the proof.
\end{proof}
Let us continue by showing how Theorem $\ref{theo: main}$ implies Theorem $\ref{theo: first}$.
\begin{proof}[Proof of Theorem $\ref{theo: first}$] From Theorem $\ref{theo: main}$, one can write $$Z_{\mathcal{I}}(G,x)=\frac{1}{2^g}\bigg|\sum_{q\in \mathcal{Q}(\Sigma)} (-1)^{\text{Arf}(q)+\text{Arf}(q_0)} \text{Pf}(A^{K_q}(G^T,x^T)) \bigg|.$$
	By part $(ii)$ of Lemma $\ref{lem: Arf}$, for each quadratic form $q$, there exists an element $\Delta_q\in H_1(\Sigma;\mathbb{Z}_2)$ such that $\text{Arf}(q)+\text{Arf}(q_0)=q_0(\Delta_q)$. This element $\Delta_q$ is determined by its Poincaré dual $\Delta_q^*\in Hom(H_1(\Sigma;\mathbb{Z}_2);\mathbb{Z}_2)$ given by $\Delta_q^*(\alpha)=q(\alpha)+q_0(\alpha)$ for every $\alpha\in H_1(\Sigma;\mathbb{Z}_2)$. Moreover, the correspondence $q\mapsto \Delta_q$ is bijective, thus we can write $$Z_{\mathcal{I}}(G,x)=\frac{1}{2^g}\bigg|\sum_{\Delta\in H_1(\Sigma;\mathbb{Z}_2)} (-1)^{q_0(\Delta)} \text{Pf}(A^{K_{q_0+\Delta^*}}(G^T,x^T)) \bigg|.$$
	Now recall that, fixing the basis $\mathcal{B}=\{[a_1],\dots,[a_g],[b_1],\dots,[b_g]\}$ of $H_1(\Sigma;\mathbb{Z}_2)$, each element $\Delta\in H_1(\Sigma;\mathbb{Z}_2)$ can be written as $\Delta=\sum_{i=1}^g\epsilon_i[a_i]+\sum_{i=1}^g\epsilon'_i[b_i]$, and thus $\Delta$ can be identified with $(\epsilon,\epsilon')\in \mathbb{Z}_2^{2g}$. Since $q_0([a_i])=q_0([b_i])=0$ and $[a_i]\cdot[b_j]=\delta_{ij}$ while $[a_i]\cdot[a_j]=[b_i]\cdot[b_j]=0$ for every $1\le i,j\le g$, by definition of quadratic forms we get $$q_0(\Delta)=\sum_{i=1}^g \epsilon_i\epsilon'_i.$$
	Therefore, the orientation $K_{q_0+\Delta^*}$, which is obtained by inverting $K_0=K_{q_0}$ on every edge $e$ such that $q_0([e])+\Delta^*([e])\neq q_0([e])$, or equivalently on every edge $e$ such that $1=\Delta^*([e])=\Delta\cdot[e]$, can be understood as $K_{\epsilon,\epsilon'}$ defined before. This fact together with the two previous equations leads to the formula stated in Theorem $\ref{theo: first}$.
\end{proof}
Now we are only left with the proof of Proposition  $\ref{pro: main1}$.
\begin{figure}[t]
	\centering
	\includegraphics[height=90pt]{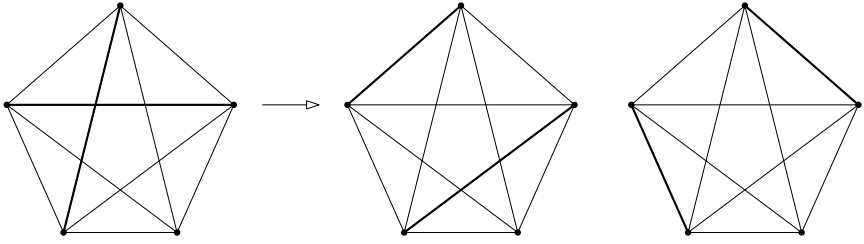}
	\caption{Reduce the number of self-intersections of a dimer configuration.}\label{fig: transformation}
\end{figure}
\subsection{Proof of Proposition $\ref{pro: main1}$}\label{subsec: main proposition Ising}
Let us recall that we need to prove the following equality 
\begin{equation} \label{eq: main}
\epsilon^{K_0}(D)=\epsilon_0(-1)^{t(D)}
\end{equation}
with $\epsilon_0=\pm 1$ a constant. To do that, we first reduce this equation to the case where the dimer configuration $D$ of $G^T$ has no inside self-intersections. More precisely, if $D$ has a self-intersection created by two short edges, let us replace these two by new ones (there are two ways to do so) as in Figure $\ref{fig: transformation}$ to obtain a new dimer configuration $D'$. It is obvious that $\trong(D)$ and $\trong(D')$ have opposite parity, and $\ngoai(D)=\ngoai(D')$. Furthermore, by the choice of $K_0$ on short edges of $G^T$ (recall Definition $\ref{def: Kasteleyn}$, condition $(i)$), one can verify easily that $\epsilon^{K_0}(D)=-\epsilon^{K_0}(D')$. By repeating this transformation (each time one decreases $\trong(D)$), one can finally obtain $D'\in \mathcal{D}(G^T)$ with $\trong(D')=0$, $\ngoai(D')=\ngoai(D)$ and $$\epsilon^{K_0}(D)=(-1)^{\trong(D)}\epsilon^{K_0}(D').$$ Hence Equation (\ref{eq: main}) now is equivalent to 
$$\epsilon^{K_0}(D')=\epsilon_0(-1)^{\ngoai(D')}.$$
Moreover, since $\trong(D')=0$, we can remove some interior short edges inside complete graphs to obtain $G'$ so that $G'$ has no more crossings inside $\mathcal{P}$ and $D'$ is still a dimer configuration of $G'$. Additionally, $K_0$ restricted on $G'$ is still a good orientation by Remark \ref{rem: good orientations}. Therefore, Proposition $\ref{pro: main1}$ boils down to the following one.
\begin{proposition} \label{pro: reduce}
	Let $G'$ be a subgraph of the terminal graph $G^T$ with no inside crossings, and $K$ a good orientation of $G'$. Then for every dimer configuration $D$ of $G'$ we have $$\epsilon^{K}(D)=\epsilon_0(-1)^{t(D)},$$
	where $\epsilon_0=\pm 1$ is a constant.
\end{proposition}
\begin{proof} First of all, until the end of this section, without stating explicitly, all equations and equalities will be understood in $\mathbb{Z}_2$. Note that $D_0$ is still the standard dimer configuration of $G'$, and recall that the symmetric difference $D\Delta D_0$ is a vertex-disjoint union of cycles $C_j$'s whose edges are long and short alternatively. Recall also that by Equation (\ref{eq: epsilon}) one has $$\epsilon^K(D_0)\epsilon^K(D)=(-1)^{\sum_j (n^K(C_j)+1)},$$ where $n^K(C_j)$ denotes the number of edges of $C_j$ where $K$ is different from a fixed orientation on $C_j$. We have the following result  \begin{equation} \label{eq: Tesler}
	\sum\limits_j (n^K(C_j)+1)=t(D)+t(D_0),
	\end{equation} whose proof will be given later. Now let us finish the proof of Proposition $\ref{pro: reduce}$ as well as Proposition \ref{pro: main1} by using this equation. Indeed, with Equation $(\ref{eq: Tesler})$ one gets
	$$\epsilon^K(D)=\epsilon^K(D_0)(-1)^{\sum_j (n^K(C_j)+1)}=\epsilon^K(D_0)(-1)^{t(D)+t(D_0)}=\epsilon_0 (-1)^{t(D)}$$ with $\epsilon _0:=\epsilon^K(D_0)(-1)^{t(D_0)}$. 
\end{proof}
In the rest of this section, let us prove Equation ($\ref{eq: Tesler}$), which is in fact a particular case of the main result shown in \cite{Tes2000}. However, in our context, with the standard dimer configuration $D_0$ the proof in \cite{Tes2000} can be simplified extremely so that we get a very elementary one, as we will show now.

Let us start by recalling that all the self-intersections (or simply crossings) of $D$ now are only created by its outside edges, that is, the ones lying outside the polygon $\mathcal{P}$, and that $D_0$ consists of all the long edges, thus contains all the outside edges. Hence we can rewrite the right-hand side of Equation ($\ref{eq: Tesler}$) as $$t(D)+t(D_0)=(D_0\setminus D)\cdot (D_0\setminus D)+D\cdot (D_0\setminus D).$$ Indeed, each self-intersection of $D$ is also a self-intersection of $D_0$, therefore $t(D)+t(D_0)$ has the same parity as the number of crossings created by either 2 outside edges not in $D$, or by one outside edge in $D$ and one not in $D$. The corresponding crossings contribute to $(D_0\setminus D)\cdot (D_0\setminus D)$ and $D\cdot (D_0\setminus D)$ respectively. Now we are left with showing \begin{equation} \label{eq: final}
\sum\limits_j (n^K(C_j)+1)=(D_0\setminus D)\cdot (D_0\setminus D)+D\cdot (D_0\setminus D).
\end{equation} This equation will be proved in two steps: we start with the case where all $C_j$'s are simple; if some of them are not, we can come back to the previous case by smoothing their self-intersections. Let us now treat the first case.
\begin{lemma} \label{lem: simple cycles}
	Equation $(\ref{eq: final})$ holds if all the cycles $C_j$'s are simple.
\end{lemma}
In order to prove this lemma, we need the following result.
\begin{claim}
	If $C$ is simple, then $n^K(C)+1$ has the same parity as $V_{\text{int}}$, the number of vertices of $G'$ inside $C$.
\end{claim}
With this claim, one can prove Lemma $\ref{lem: simple cycles}$ as follows. 
\begin{proof} [Proof of Lemma $\ref{lem: simple cycles}$]
	In fact we will show that the contributions of each $C_j$ to both sides of $(\ref{eq: final})$ are equal, or equivalently, the number of crossings created by outside edges of $C_j$ with $D$ or $D_0\setminus D$ has the same parity as the number of vertices inside $C_j$. Indeed, let $v$ be a vertex inside $C_j$, since $G'\setminus C_j$ admits a dimer configuration, $v$ is matched to another vertex $v'$ by a dimer of this configuration. As $v'$ is not on $C_j$, there are only two cases: either $v'$ is inside $C_j$, or $v'$ is outside. In the first case, whatever the dimer matching $v$ and $v'$ is a long or short edge, it must intersect $C_j$ an even number of times (remember that $C_j$ is drawn on the plane). Hence $v$ (or the pair $(v,v')$) contributes $0\in \mathbb{Z}_2$ to both $V_{\text{int}}$ and the quantity $(D_0\setminus D)\cdot (D_0\setminus D)+D\cdot (D_0\setminus D)$. In the latter case, the dimer matching $v$ and $v'$ must intersect $C_j$ an odd number of times, and so it is a long edge. Depending on whether this long edge belongs to $D$ or $D_0\setminus D$, its crossings with outside edges of $C_j$ (which are all long edges, and so belong to $D_0\setminus D$ since $C_j$ contains long and short edges alternatively) contribute to $D\cdot (D_0\setminus D)$ or $(D_0\setminus D)\cdot (D_0\setminus D)$, but not both. Therefore, in this case $v$ contributes $1\in \mathbb{Z}_2$ to both $V_{\text{int}}$ and $(D_0\setminus D)\cdot (D_0\setminus D)+D\cdot (D_0\setminus D)$. 
\end{proof}
Now let us complete the proof of Lemma $\ref{lem: simple cycles}$ by proving the claim.
\begin{proof} [Proof of the claim] First of all, set $m$ to be the number of outside edges belonging to $C$. Since $C$ is a simple closed curve drawn on the plane, it must bound a disc. Let us consider the case where this disc is inside $\mathcal{P}$, that is, $m=0$. In fact this case is the same as Kasteleyn's theorem displayed in Section \ref{sec: Kasteleyn}: using the fact that $K$ is good (recall Definition \ref{def: Kasteleyn}, condition $(ii)$), we get the statement in the claim.
	\begin{figure}[t]
		\centering
		\includegraphics[height=200pt]{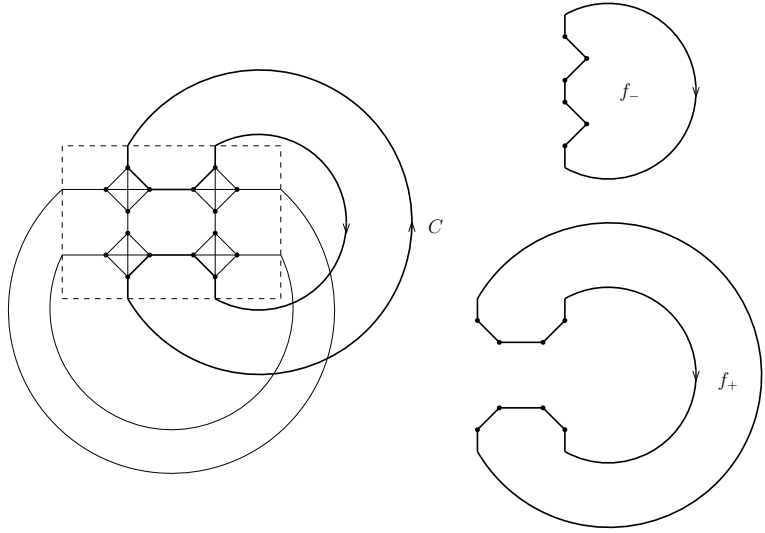}
		\caption{The case $m=2$: a simple cycle $C$ gives rise to an outside positive face $f_+$ and an outside negative face $f_-$.}\label{fig: m=2}
	\end{figure}  
	
	Let us now continue with the case $m\geq 1$. Note that if $m=1$, one can verify easily that the argument above works completely well with an additional fact: the orientation $K$ satisfies condition $(iii)$ of Definition $\ref{def: Kasteleyn}$ for the unique outside edge of $C$. This is also the case for general $m$ if no two of the $m$ outside edges are homologous. The situation only becomes more intricate when we have some outside edges of same homology classes. This is the content of what we will show next.
	
	Recall that to each outside edge, there is a unique face, which is called the outside face, formed by this edge and some edges along the boundary of the subgraph obtained from $G'$ by removing all the other outside edges (cf. Definition $\ref{def: Kasteleyn}$). Let us consider the $m$ outside faces corresponding to $m$ outside edges of $C$. Note that when we travel along $C$ with the counterclockwise orientation, some outside faces have boundaries partly travelled with the same orientation, while the others have boundaries partly travelled with the opposite one. Let us call the former ones $\emph{positive}$ and the latter ones $\emph{negative}$ (see Figure $\ref{fig: m=2}$ for example). For inside faces of the disc bounded by $C$, we also call them $\emph{positive}$. With these terminology and the same notations as before, using conditions $(ii)$ and $(iii)$ of Definition $\ref{def: Kasteleyn}$ one gets \begin{eqnarray*}
		0&=&\sum_{f\, \text{positive}}(n^K(\partial f)+1)+\sum_{f\, \text{negative}}(n^K(\partial f)+1)\\
		&=& \sum_{f\, \text{positive}}1+\sum_{f\, \text{positive}}n^K(\partial f)+\sum_{f\, \text{negative}}(n^K(-\partial f)+|\partial f|+1)\\
		&=&F+\big(\sum_{f\, \text{positive}}n^K(\partial f)+\sum_{f\, \text{negative}}n^K(-\partial f)\big)+\sum_{f\,\text{negative}}(|\partial f|+1)\\
		&=&F+\big(n^K(C)+E_{\text{int}}+\sum_{f\,\text{negative}}(|\partial f|-1)\big)+\sum_{f\,\text{negative}}(|\partial f|+1)\\
		&=&F+n^K(C)+E_{\text{int}}\\
		&=&n^K(C)+V_{\text{int}}+1,
	\end{eqnarray*} which implies the claim.
	In this equation, the third equality comes from the fact that each positive face contributes 1 to $F$, while negative ones do not contribute. The forth equality can be explained as follows: each edge belonging to $\partial f$ (resp. $-\partial f$) for $f$ positive (resp. negative) that is not a common edge of any two faces (among faces bounded by $C$) contributes to $n^K(C)$; moreover, common edges of two positive faces contribute to $E_{\text{int}}$, while common edges of a positive face and a negative face contribute to the length of that negative face minus 1. Finally, the last equality comes from the argument for $m=0$ shown before. We have done with the proof of the claim, as well as Lemma $\ref{lem: simple cycles}$.
\end{proof}
We now finish the proof of Equation (\ref{eq: final}) by considering the remaining case where some of the $C_j$'s are not simple.
\begin{lemma}
	Equation $(\ref{eq: final})$ holds for every collection of cycles $\{C_j\}$.
\end{lemma}
\begin{proof}
	The idea of the proof is to transform the collection of cycles $\{C_j\}$ together with the good orientation $K$ to have new ones so that we can apply Lemma $\ref{lem: simple cycles}$. Let us start by supposing that there is only one cycle $C$ of this collection which is not simple, and that $C$ has only one self-intersection. Remember that since this self-intersection is created by two outside edges (which are also long edges and so belong to $D_0\setminus D$), it contributes to the quantity $(D_0\setminus D)\cdot (D_0\setminus D)$. Smoothing $C$ at its self-intersection as in Figure $\ref{fig: smooth}$ we create two new outside edges splitting $C$ into two simple closed curves $C'$ and $C''$. Moreover, doing so introduces a new standard dimer configuration $D'_0$, but still keeps $D$ the same on the resulting graph. It is clear that $C'$ and $C''$ still have alternative edges between $D$ and $D'_0\setminus D$, and so have even lengths.
	\begin{figure}[b]
		\centering
		\includegraphics[height=90pt]{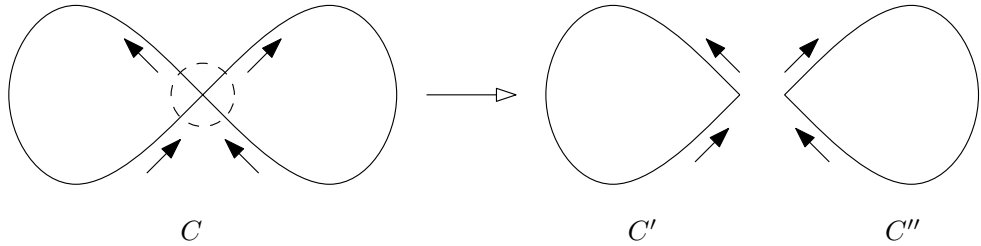}
		\caption{Smooth $C$ at its self-intersection.}\label{fig: smooth}
	\end{figure}  
	\begin{figure}[t]
		\centering
		\includegraphics[height=140pt]{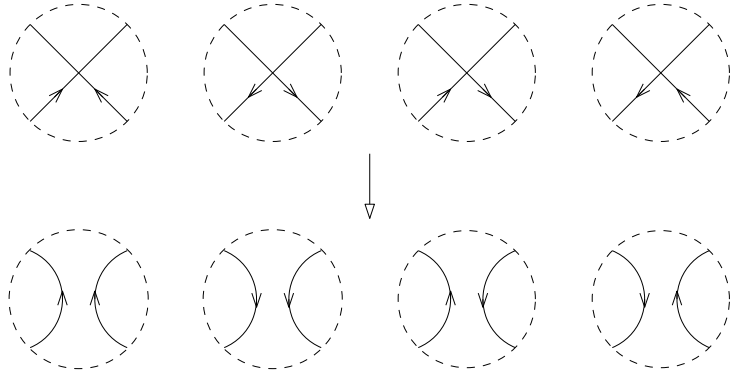}
		\caption{Transform $K$ at the self-intersection of $C$.}\label{fig: smooth K}
	\end{figure} 
	
	Next we transform the orientation $K$ to $K'$ as in Figure $\ref{fig: smooth K}$, in which we have a certain flexibility. More precisely, one can switch two choices of $K'$ in the first two cases, as well as those in the last two ones. With all these choices of $K'$, one clearly has $$n^K(C)=n^{K'}(C')+n^{K'}(C'').$$ This implies that replacing the cycle $C$ by the pair $(C',C'')$ (resulting to replace the collection $\{C_j\}$ by a new one of all simple cycles whose edges are still long and short alternatively) and the orientation $K$ by $K'$ changes the parity of both sides of Equation $(\ref{eq: final})$, and hence does not affect to the validity of this equation. To apply Lemma $\ref{lem: simple cycles}$ for this new collection of cycles and the orientation $K'$, one only needs to verify if $K'$ is good on the new graph obtained by smoothing $C$. Recall that by smoothing $C$ we introduce two new outside edges which give rise to two new outside faces, say $f_1,f_2$. We only need to check that $K'$ satisfies condition $(iii)$ of Definition $\ref{def: Kasteleyn}$ on these two new faces. Indeed, let us denote by $f_3,f_4$ the two outside faces corresponding to the two outside edges that create the self-intersection of $C$. By definition of $K$ and $K'$, it follows that $$0=n^K(\partial f_3)+n^K(\partial f_4)=n^{K'}(\partial f_1)+n^{K'}(\partial f_2)+2n^K(l)=n^{K'}(\partial f_1)+n^{K'}(\partial f_2),$$ where $l$ is the common path of $\partial f_3$ and $\partial f_4$. Now if both $n^{K'}(\partial f_1)$ and $n^{K'}(\partial f_2)$ are odd, we are done. If both of them are even, we only need to invert $K'$ on the two newly created outside edges, using the flexibility mentioned above.
	
	If $C$ has many self-intersections, we smooth them one by one, and at each step, we transform $K$ to $K'$ as above. If $\{C_j\}$ has many cycles which are not simple, we treat them one by one as in the previous case. This concludes the proof of Equation $(\ref{eq: final})$, as well as Proposition $\ref{pro: reduce}$ and Theorem $\ref{theo: main}$.
\end{proof}
\section{The non-orientable case} \label{sec: non-orientable case Ising}
The method we have developed to calculate the Ising partition function for graphs embedded in orientable surfaces can be extended to the case of non-orientable surfaces with some slight modifications. This section is spent to show our Pfaffian formulas in this case. Similarly to the orientable case we obtain a practical version of the Pfaffian formula (see Theorem $\ref{theo: second}$ below), which is an extension of Theorem $\ref{theo: first}$ to non-orientable surfaces. Also, using quadratic enhancements and the Brown invariant, which are generalisations of quadratic forms and the Arf invariant to possibly non-orientable surfaces, we get a theoretical version of the Pfaffian formula (Theorem $\ref{theo: general}$) which generalises Theorem $\ref{theo: main}$. As a consequence, when we restrict our general Pfaffian formulas to the orientable case, we get back the ones stated before. Let us now go into detail with some new terminology and then state our general Pfaffian formulas. Their proofs will be left until the end of this section.
\subsection{Statement of the general Pfaffian formulas} Throughout this section, let us assume that the graph $G$ is embedded in a non-orientable surface $\Sigma$. We first need to fix a representation of $\Sigma$ as follows. If $\Sigma$ is homeomorphic to the connected sum of an orientable surface $\Sigma_g$ of genus $g$ with the Klein bottle $\mathscr{K}$, we represent $\Sigma$ as a polygon $\mathcal{P}$ with sides identified following the word $a_1b_1a_1^{-1}b_1^{-1}\cdots a_gb_ga_g^{-1}b_g^{-1}aba^{-1}b$; if $\Sigma$ is homeomorphic to the connected sum of $\Sigma_g$ with the projective plane $\mathbb{R}P^2$, we represent $\Sigma$ as a polygon $\mathcal{P}$ with sides identified following the word $a_1b_1a_1^{-1}b_1^{-1}\cdots a_gb_ga_g^{-1}b_g^{-1}cc$. Note that we still have the identification $H_1(\Sigma;\mathbb{Z}_2)\equiv H_1(\Sigma,\mathcal{P};\mathbb{Z}_2)$ as in the orientable context. The fixed drawing of $G$ as well as of $G^T$ can be defined in the same way as in the orientable case.

Secondly, for further purposes, let us define a function $\omega:E(G)\to\Z_2$ as follows. If $\Sigma$ is represented by the word $a_1b_1a_1^{-1}b_1^{-1}\cdots a_gb_ga_g^{-1}b_g^{-1}aba^{-1}b$ (resp. $a_1b_1a_1^{-1}b_1^{-1}\cdots a_gb_ga_g^{-1}b_g^{-1}cc$), for each $e\in E(G)$, we define $\omega(e)$ as the parity of the number of intersections of $e$ with side $b$ (resp. with side $c$). Note that this function can be extended naturally to a function (that we still denote by $\omega$) on the edges of $G^T$.

Next we will define orientations of interest. Clearly one still can define good orientations by Definition $\ref{def: Kasteleyn}$ as before, and Proposition $\ref{pro: Kasteleyn}$ holds in this new representation of $\Sigma$. Let us pick a good orientation $K_0$ and show how to derive other particular orientations from it. Recall that if $\Sigma=\Sigma_g\# \mathscr{K}$ the word representing $\Sigma$ induces the basis $\mathcal{B}':=\{[a_1],\dots,[a_g],[b_1],\dots,[b_g],[a],[b]\}$ of $H_1(\Sigma;\mathbb{Z}_2)$. Then each element $\Delta$ of $ H_1(\Sigma;\mathbb{Z}_2)$ can be written as $$\Delta=\sum_{i=1}^g\epsilon_i[a_i]+\sum_{i=1}^g\epsilon'_i[b_i]+\alpha[a]+\beta[b],$$ and hence can be identified with $(\epsilon,\epsilon',\alpha,\beta)\in \mathbb{Z}_2^{b_1}$. For such $(\epsilon,\epsilon',\alpha,\beta)\in \mathbb{Z}_2^{b_1}$, let us denote by $K_{\epsilon,\epsilon',\alpha,\beta}$ the orientation obtained by inverting $K_0$ on every edge $e$ each time $e$ crosses a side $a_i$ of $\mathcal{P}$ (resp. $b_j$) with $\epsilon_i=1$ (resp. $\epsilon'_j=1$) as well as the side $a$ (resp. $b$) with $\alpha=1$ (resp. $\beta=1$). The set $\{K_{\epsilon,\epsilon',\alpha,\beta}: (\epsilon,\epsilon',\alpha,\beta)\in \mathbb{Z}_2^{2g+2}\}$ consists of orientations that we are interested in. Similarly if $\Sigma=\Sigma_g\# \mathbb{R}P^2$ we have the basis $\mathcal{B}'':=\{[a_1],\dots,[a_g],[b_1]\dots,[b_g],[c]\}$ of $H_1(\Sigma;\mathbb{Z}_2)$, and each $\Delta\in H_1(\Sigma;\mathbb{Z}_2)$ can be identified with $(\epsilon,\epsilon',\gamma)\in \mathbb{Z}_2^{2g+1}$. Then the orientations $K_{\epsilon,\epsilon',\gamma}$ of interest are obtained by inverting $K_0$ on every edge $e$ each time $e$ crosses a side $a_i$ of $\mathcal{P}$ (resp. $b_j$) with $\epsilon_i=1$ (resp. $\epsilon_j'=1$) as well as the side $c$ with $\gamma=1$.

Finally, we need a twisted version of adjacency matrices that can be defined as follows (cf.\cite{Cim09}). Let $(G,x)$ be an edge-weighted graph with $2n$ vertices labelled by $\{1,\dots,2n\}$.  Suppose that $K$ is an arbitrary orientation on the edges of $G$ and $w$ is a function on $E(G)$. The \emph{twisted adjacency matrix} of $G$ with respect to $K$ and $\omega$, denoted by $A^{K,\omega}(G)$, has entries given by
\begin{equation}\label{eq: twisted adj matrix}
a_{ij}=\sum_{e=(i,j)}\epsilon^K_{ij}(e)i^{\omega(e)}x(e),
\end{equation}
where the sum is taken over all the edge $e$ of $G$ connecting vertices $i,j$, and $\epsilon^K_{ij}(e)$ is defined as before.
\smallskip

We are now ready to state the hands-on version of the general Pfaffian formula.
\begin{theorem}
	\label{theo: second}
	Let $(G,x)$ be a weighted graph embedded in the non-orientable surface $\Sigma$. Then the Ising partition function of $G$ is given by
	$$Z_{\mathcal{I}}(G,x)=\frac{1}{2^{g+1}}\bigg| \sum_{(\epsilon,\epsilon',\alpha,\beta)\in \mathbb{Z}_2^{2g+2}} (-1)^{\sum_{i=1}^g \epsilon_i\epsilon_i'+\alpha\beta}(-i)^{\alpha}\text{Pf}(A^{K_{\epsilon,\epsilon',\alpha,\beta},\omega}(G^T,x^T))\bigg|$$ if $\Sigma=\Sigma_g\#\mathscr{K}$, and by 
	$$Z_{\mathcal{I}}(G,x)=\frac{1}{2^{g+1/2}}\bigg| \sum_{(\epsilon,\epsilon',\gamma)\in \mathbb{Z}_2^{2g+1}} (-1)^{\sum_{i=1}^g \epsilon_i\epsilon_i'}\,i^{\gamma}\,\text{Pf}(A^{K_{\epsilon,\epsilon',\gamma},\omega}(G^T,x^T))\bigg|$$ if $\Sigma=\Sigma_g\#\mathbb{R}P^2$. In the above formulas, $A^{K_{\epsilon,\epsilon',\alpha,\beta},\omega}(G^T,x^T)$ and $A^{K_{\epsilon,\epsilon',\gamma},\omega}(G^T,x^T)$ are twisted adjacency matrices of the terminal graph $(G^T,x^T)$ with respect to $\omega$ and the orientations $K_{\epsilon,\epsilon',\alpha,\beta}$, $K_{\epsilon,\epsilon',\gamma}$ respectively.
\end{theorem}

Note that if $\Sigma=\Sigma_g$ is an orientable surface of genus $g$, we can omit the indices $\alpha,\beta$ (resp. the index $\gamma$) in the first formula (resp. in the second) as well as $\omega$ (since $\omega\equiv 0$ then). We can also replace $g+1$ (resp. $g+1/2$) by $g$ (which is natural since they all indicate $\dfrac{1}{2}\dim H_1(\Sigma;\mathbb{Z}_2)$, see the proof below), thus we get back the formula in Theorem $\ref{theo: first}$.

Before giving another version of Theorem \ref{theo: second}, let us consider a simple example to see how the Pfaffian formula works.
\begin{figure}
	\centering
	\includegraphics[height=120pt]{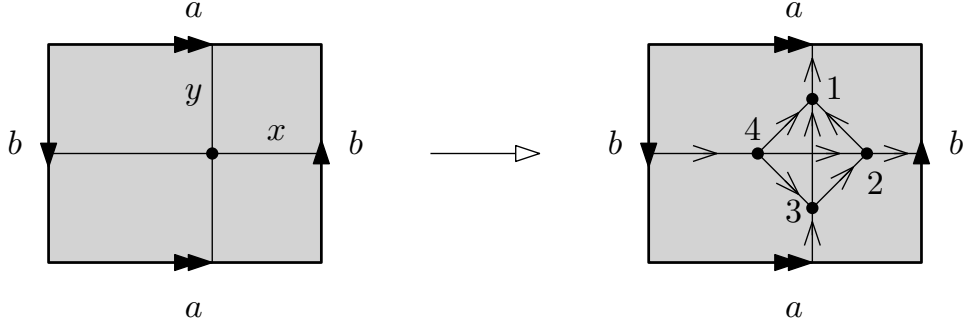}
	\caption{The graph $G\subset\K$ and its associated terminal graph $G^T$ with a good orientation $K$.\label{fig: exampleKlein1}}
\end{figure}
\begin{example}
	Let $G$ be the $1\times 1$ square lattice embedded in the Klein bottle $\mathscr{K}$ with horizontal weight $x$ and vertical weight $y$. Label the vertices of $G^T$ as in Figure \ref{fig: exampleKlein1}, and choose $K$ a good orientation. By definition we have $$A^{K}(G^T,x^T)=\begin{pmatrix} 0 & -\sqrt{xy} & 1-y & -\sqrt{xy}\\
	\sqrt{xy} & 0 & -\sqrt{xy} & i-x \\
	-1+y & \sqrt{xy} & 0 & -\sqrt{xy} \\
	\sqrt{xy} & -i+x & \sqrt{xy} & 0
	\end{pmatrix}$$ whose Pfaffian is $\Pf(A^{K}(G^T,x^T))=\Pf(A^{K_{0,0}}(G^T,x^T))=xy-i+x+iy$. Similarly one can find $$\Pf (A^{K_{1,0}}(G^T,x^T))=xy+i-x+iy,$$ $$\Pf (A^{K_{0,1}}(G^T,x^T))=xy+i+x-iy,$$ 
	$$\Pf (A^{K_{1,1}}(G^T,x^T))=xy-i-x-iy.$$
	 Then, dropping $\omega$ for simplicity, the formula in Theorem \ref{theo: second} gives
	\begin{multline*}
	Z_{\mathcal{I}}(G,x)=\frac{1}{2}\big|\Pf (A^{K_{0,0}}(G^T,x^T))+\Pf (A^{K_{0,1}}(G^T,x^T))\\-i\,\Pf (A^{K_{1,0}}(G^T,x^T))+i\,\Pf (A^{K_{1,1}}(G^T,x^T)) \big|,\end{multline*} leading to $Z_{\mathcal{I}}(G,x)=xy+x+y+1$ as we found before.
\end{example}
\smallskip

Let us now continue with the alternative version of the general Pfaffian formula. To do so, we need the terminology of  quadratic enhancements and Brown invariant that we now briefly recall. Readers are referred to \cite{KirTay90} for more details.

Suppose that $(H,\cdot)$ is a finite dimensional vector space over $\mathbb{Z}_2$ together with a symmetric bilinear form $(\cdot):H\times H\rightarrow \mathbb{Z}_2$. Given a linear function $\omega:H\rightarrow \mathbb{Z}_2$, a function $q: H\rightarrow \mathbb{Z}_4$ is called a $\emph{quadratic enhancement}$ \index{quadratic enhancement} of $(H,\cdot)$ with respect to $\omega$ if for every $x,y\in H$ we have $q(x)-\omega(x)\in 2\Z_2$, and $$q(x+y)=q(x)+q(y)+2(x\cdot y),$$ where $2:\mathbb{Z}_2\rightarrow \mathbb{Z}_4$ denotes the inclusion homomorphism. Note that a quadratic enhancement is completely determined by its values on a basis of $H$, and these values can be chosen freely according to the condition $q(x)-\omega(x)\in 2\Z_2$. If the bilinear form $(\cdot)$ is non-degenerate, the \emph{Brown invariant} $\text{Br}(q)$  of a quadratic enhancement $q$ is an integer modulo 8 defined by $$\exp\left(\frac{i\pi}{4}\right)^{\text{Br}(q)}=\frac{1}{\sqrt{|H|}}\sum_{x\in H}i^{q(x)}.$$ 
\smallskip

Coming back to our context, the function $\omega:E(G)\to\Z_2$ defined before induces a $\Z_2$-valued function on the set of cycles of $G$ (that is, if $C$ is a cycle then we can define $\omega(C):=\sum_{e\in C}\omega(e)$). Hence it turns out that $\omega$ induces a linear $\Z_2$-valued function on $H_1(\Sigma;\Z_2)$ which represents the first Stiefel$-$Whitney class of $\Sigma$ (cf. \cite[Section 4]{Cim09}). This class characterises the orientability of $\Sigma$ and it is of course identically $0$ if $\Sigma$ is orientable. By abuse of notation, we still denote this class by $\omega$ and focus on quadratic enhancements on $(H_1(\Sigma;\Z_2),\cdot)$ with respect to $\omega$. Denote the set of all such enhancements by $\mathcal{Q}(\Sigma,\omega)$.

We now define a particular quadratic enhancement $\tilde{q}_0\in \mathcal{Q}(\Sigma,\omega)$ as follows.  If $\Sigma=\Sigma_g\#\K$, $\tilde{q}_0$ is given by its values on the basis $\mathcal{B}'$ of $H_1(\Sigma;\Z_2)$ by $$\tilde{q}_0([a_i])=\tilde{q}_0([b_i])=\tilde{q}_0([b])=0\quad \text{and}\quad \tilde{q}_0([a])=3.$$ If $\Sigma=\Sigma_g\#\R P^2$, we define $\tilde{q}_0$  by setting $$\tilde{q}_0([a_i])=\tilde{q}_0([b_i])=0\quad \text{and}\quad \tilde{q}_0([c])=1$$ on the basis $\mathcal{B}''$. Then by fixing a good orientation $K_0$, for each quadratic enhancement $q\in \mathcal{Q}(\Sigma,\omega)$ we denote by $K_q$ the orientation obtained by inverting $K_0$ on every edge $e$ such that $q([e])\neq \tilde{q}_0([e])$. We now can state the general version of the Pfaffian formula.
\begin{theorem}
	\label{theo: general}
	Let $(G,x)$ be a weighted graph embedded in a possibly non-orientable surface $\Sigma$. Then the Ising partition function on $G$ is given by $$Z_{\mathcal{I}}(G,x)=\frac{\epsilon_0}{2^{b_1/2}}\sum_{q\in \mathcal{Q}(\Sigma,\omega)}\exp \left(\frac{i\pi}{4}\right)^{-\text{Br}(q)} \text{Pf}(A^{K_q,\omega}(G^T,x^T)),$$ where $\epsilon_0=\pm 1$ is a constant and $b_1:=\dim H_1(\Sigma;\mathbb{Z}_2)$. In this formula, $A^{K_q,\omega}(G^T,x^T)$ is the twisted adjacency matrix of the terminal graph $(G^T,x^T)$ with respect to the orientation $K_q$ and the standard representation $\omega$ of $w_1$.
\end{theorem}
Observe that if $\Sigma$ is orientable then $\omega\equiv 0$ and a quadratic enhancement $q\in\mathcal{Q}(\Sigma,\omega)$ induces a quadratic form $\dfrac{q}{2}$ whose Arf invariant is equal to $\dfrac{1}{4}\,\text{Br}(q)$. So the formula in Theorem \ref{theo: general} restricted to orientable surfaces is the one stated in Theorem $\ref{theo: main}$.
\subsection{Proof of Theorems $\ref{theo: second}$ and $\ref{theo: general}$} The proofs of Theorems $\ref{theo: second}$ and $\ref{theo: general}$ are similar to those of Theorems $\ref{theo: first}$ and $\ref{theo: main}$, with some slight differences. However, our key result in the previous section, Proposition $\ref{pro: main1}$, is still valid in the non-orientable context. As in the previous section, let us go through new terminology and definitions once again to give more details and results that are needed, before giving our proofs.

First of all, Proposition $\ref{pro: outside intersections}$ can be generalised to the following one.
\begin{proposition} \label{pro: general intersection}
	For every $M\subset E(G)$, the number $t(M)$ of its self-intersections in the plane has the same parity as the number $\dfrac{1}{2}(\tilde{q}_0([M])-\omega(M))$ where $[M]\in H_1(\Sigma,\mathcal{P};\mathbb{Z}_2)$ is the relative homology class of $M$.
\end{proposition}
\begin{proof}
	Note that Lemma $\ref{lem: intersection}$ is still valid when $\Sigma$ is non-orientable. Using the fact that $[a_i]\cdot[b_j]=\delta_{ij}$ and $[a_i]\cdot[a_j]=[b_i]\cdot[b_j]=0$ as before, together with $[a]\cdot[b]=1$ and $[a]\cdot [a_i]=[a]\cdot[b_i]=[b]\cdot[a_i]=[b]\cdot [b_i]=0$ if $\Sigma=\Sigma_g\#\mathcal{K}$, while $[c]\cdot[c]=1$ and $[c]\cdot[a_i]=[c]\cdot[b_i]=0$ if $\Sigma=\Sigma_g\#\mathbb{R}P^2$, the proof of Proposition $\ref{pro: outside intersections}$ can be extended.
\end{proof}
We will also need the two following properties of Brown invariants.
\begin{lemma}\label{lem: Brown}
	Given $(H,\cdot,w)$ a finite dimensional vector space over $\mathbb{Z}_2$ together with a non-degenerate bilinear form $(\cdot)$ and a fixed linear function $w:H\rightarrow \mathbb{Z}_2$,  we have:
	\begin{enumerate}
		\item[(i)] The equality $\frac{1}{\sqrt{|H|}}\sum\limits_{q\in \mathcal{Q}(H,\cdot,w)}\exp(i\pi/4)^{-\text{Br}(q)}i^{q(x)}=1$ holds for every $x\in H$, where the sum is taken over all the set of quadratic enhancements;
		\item[(ii)] If $q_1,q_2$ are quadratic enhancements on $(H,\cdot,\omega)$, then $\text{Br}(q_1)-\text{Br}(q_2)=2q_1(\Delta)=2q_2(\Delta)$, where $\Delta\in H$ satisfies $q_1(x)+2(\Delta \cdot x)=q_2(x)$ for every $x\in H$.
	\end{enumerate}
\end{lemma}
Recall that for $\omega\equiv 0$, the quadratic form $q'=\dfrac{q}{2}$ induced from $q$ has $\text{Arf}(q')=\dfrac{1}{4}\,\text{Br}(q)$. Hence Lemma \ref{lem: Brown} above implies Lemma \ref{lem: Arf} stated before.
\begin{proof}[Proof of Lemma \ref{lem: Brown}] For the first part, let us fix a quadratic enhancement $q_*$. For each $z\in H$, 
	we define a function $q_z:H\rightarrow \mathbb{Z}_4$ by $q_z(x):=q_*(x)+2(z\cdot x)$. It is clear by definition that $q_z$ is also a quadratic enhancement.
	With this notation we have
	$$\sum_{x\in H} i^{q_z(x)} = \sum_{x\in H} i^{q_*(x)+2(z\cdot x)}
	= \sum_{x\in H} i^{q_*(x+z)-q_*(z)}
	= i^{-q_*(z)}\sum_{x\in H} i^{q_*(x+z)}
	= i^{-q_*(z)}\sum_{x\in H} i^{q_*(x)},$$
	which means that $$\exp\left(\frac{i\pi}{4}\right)^{\text{Br}(q_z)}=i^{-q_*(z)}\exp\left(\frac{i\pi}{4}\right)^{\text{Br}(q_*)},$$ 
	or equivalently, $\text{Br}(q_z)+2q_*(z)=\text{Br}(q_*)$. Therefore, observing that $z \mapsto q_z$ is a bijection between $H$ and the set of quadratic enhancements on it, we get
	\begin{align*}
	\sum_{q\in \mathcal{Q}(H,\cdot,\omega)}\exp\left(\frac{i\pi}{4}\right)^{-\text{Br}(q)}i^{q(x)}
	&= \sum_{z\in H} \exp\left(\frac{i\pi}{4}\right)^{-\text{Br}(q_z)}i^{q_z(x)}
	= \sum_{z\in H} \exp\left(\frac{i\pi}{4}\right)^{-\text{Br}(q_*)+2q_*(z)}i^{q_z(x)}\\
	&= \exp\left(\frac{i\pi}{4}\right)^{-\text{Br}(q_*)}\sum_{z\in H}i^{q_z(x)+q_*(z)}= \exp\left(\frac{i\pi}{4}\right)^{-\text{Br}(q_*)}\sum_{z\in H}i^{q_*(x+z)}\\
	&= \exp\left(\frac{i\pi}{4}\right)^{-\text{Br}(q_*)}\sum_{z\in H}i^{q_*(z)}=\sqrt{|H|}.
	\end{align*}
	This concludes the proof of part $(i)$. Let us move on to the second part. By definition $q_2-q_1$ is a linear form taking even values, so there exists an element $\Delta\in H$ such that $q_2(x)-q_1(x)=2\,\Delta\cdot x$ for every $x\in H$. Thus we get $2q_1(\Delta)=2q_2(\Delta)$. Now we can write
	\begin{eqnarray*}
		\exp\left(\frac{i\pi}{4}\right)^{\text{Br}(q_1)} &=& \frac{1}{2^{b_1/2}} \sum_{x\in H}i^{q_1(x)}
		= \frac{1}{2^{b_1/2}} \sum_{x\in H}i^{q_1(x+\Delta)}= \frac{1}{2^{b_1/2}} \sum_{x\in H}i^{q_1(x)+q_1(\Delta)+2(\Delta\cdot x)}\\ &=& \frac{1}{2^{b_1/2}} \sum_{x\in H}i^{q_2(x)+q_1(\Delta)}= i^{q_1(\Delta)}\exp\left(\frac{i\pi}{4}\right)^{\text{Br}(q_2)},
	\end{eqnarray*}
	which implies that $\text{Br}(q_1)-\text{Br}(q_2)=2q_1(\Delta)$. This completes the proof of part $(ii)$.
\end{proof}
 
We continue with the transformation of the twisted partition function of the Ising model on $G$ to that of the dimer model on $G^T$. Lemma \ref{lem: general Z} below is a generalisation of Lemma $\ref{lem: Z}$.
\begin{lemma}\label{lem: general Z}
	Setting $Z^q(G,x):= \sum\limits_{\alpha\in H_1(\Sigma;\mathbb{Z}_2)}i^{q(\alpha)}\sum\limits_{[P]=\alpha}x(P)$ for each quadratic enhancement $q\in \mathcal{Q}(\Sigma,\omega)$, we have
	$$ Z^q(G,x)= \sum_{D\in \mathcal{D}(G^T)}i^{q([D\Delta D_0])+2\trong(D)}x^{T}(D),$$ where $D_0$ is the standard dimer configuration of $G^T$.
\end{lemma}
\begin{proof}
	Using the same argument and notations as in Lemma \ref{lem: Z} one can write
	\begin{align*}
		Z^q(G,x)&= \sum_{\alpha\in H_1(\Sigma;\mathbb{Z}_2)}i^{q(\alpha)}\sum_{[P]=\alpha}x(P) 
		=\sum_{\alpha\in H_1(\Sigma;\mathbb{Z}_2)}i^{q(\alpha)}\sum_{[P]=\alpha} \big( \prod_{v\in V}\sum_{D_v\in \mathcal{D}_{2n(v)}}(-1)^{\trong(D_v)} \big)x(P) \\
		&=\sum_{\alpha\in H_1(\Sigma;\mathbb{Z}_2)}i^{q(\alpha)}\sum_{[P]=\alpha}\sum_{D\in \varphi^{-1}(P)}(-1)^{\trong(D)}x^{T}(D)
		= \sum_{D\in \mathcal{D}(G^T)}i^{q([G\setminus D_G])+2\trong(D)}x^{T}(D)\\
		&= \sum_{D\in \mathcal{D}(G^T)}i^{q([D\Delta D_0])+2\trong(D)}x^{T}(D),
	\end{align*}
	which concludes our proof.
\end{proof}
Finally before going to the proofs of Theorem $\ref{theo: second}$ and $\ref{theo: general}$, let us recall that the twisted adjacency matrix $A^{K,\omega}(G,x)$ defined by Equation (\ref{eq: twisted adj matrix}) satisfies the following equation  \begin{equation} \label{eq: general Pfaffian}
\text{Pf}(A^{K,\omega}(G,x))=\sum\limits
_{D\in \mathcal{D}(G)} \epsilon^K(D)i^{\omega(D)}x(D),
\end{equation} where $\epsilon^K(D)$ is given by Equation (\ref{eq: sign}) and $\omega(D)=\sum_{e\in D}\omega(e)$.
\begin{proof}[Proof of Theorem $\ref{theo: general}$]
	First of all, as in the proof of Theorem $\ref{theo: main}$, using Equation ($\ref{eq: sign}$) one can write $ \epsilon^{K_q}(D)=\epsilon^{K_0}(D)(-1)^{|\{e\in D: q([e])\neq \tilde{q}_0([e])\}}$. Note that $q([e])$ and $\tilde{q}_0([e])$ have the same parity (both have the parity of $\omega(e)$), thus the fact that $q([e])\neq \tilde{q}_0([e])$ is equivalent to $q([e])-\tilde{q}_0([e])=2\in \mathbb{Z}_4$. Hence by definition of quadratic enhancements we can write$$ q([D])-\tilde{q}_0([D])=\sum_{e\in D}q([e])-\sum_{e\in D}\tilde{q}_0([e])=2 |\{e\in D: q([e])\neq \tilde{q}_0([e])\}|\in \Z_4,$$ which implies that $ \epsilon^{K_q}(D)=\epsilon^{K_0}(D)i^{q([D])-\tilde{q}_0([D])}$. By Proposition $\ref{pro: general intersection}$ we have $\tilde{q}_0([D])-\omega(D)=2\ngoai(D)$ modulo 4 for every $D\in \mathcal{D}(G^T)$. Using this fact and Proposition $\ref{pro: main1}$ (which is still valid in the non-orientable case) we have
	\begin{align*}
		\epsilon^{K_q}(D)&=\epsilon_0(-1)^{t(D)}i^{q([D])-\tilde{q}_0([D])}=\epsilon_0(-1)^{\trong(D)+\ngoai(D)}i^{q([D])-\tilde{q}_0([D])}\\
		&= \epsilon_0 i^{2\trong(D)+\tilde{q}_0([D])-\omega(D)}i^{q([D])-\tilde{q}_0([D])}=\epsilon_0 i^{2\trong(D)+q([D])-\omega(D)}.
	\end{align*}
	Secondly, recall that the dual $[D_0]^*$ of the relative homology class $[D_0]\in H_1(\Sigma,\mathcal{P};\mathbb{Z}_2)$ is given by $[D_0]^*([D]):=[D]\cdot [D_0]$ for every $[D]\in H_1(\Sigma,\mathcal{P};\mathbb{Z}_2)$. Since $[D_0]^*$ is linear, $q+2[D_0]^*$ is still a quadratic enhancement. Using the previous equation for  $q+2[D_0]^*$ together with Lemma $\ref{lem: general Z}$, one can write 
	\begin{align}
	Z^q(G,x)&=\sum_{D\in \mathcal{D}(G^T)}i^{q([D\Delta D_0])+2\trong(D)}x^T(D)\nonumber\\
	&=i^{q([D_0])}\sum_{D\in \mathcal{D}(G^T)}i^{q([D])+2[D]\cdot [D_0]+2\trong(D)}x^T(D)\nonumber\\  
	&=i^{q([D_0])}\sum_{D\in \mathcal{D}(G^T)}i^{(q+2D_0^*)([D])+2\trong(D)}x^T(D)\nonumber\\&=\epsilon_0 i^{q([D_0])}\sum_{D\in \mathcal{D}(G^T)}\epsilon^{K_{q+2[D_0]^*}}(D)i^{\omega(D)}x^T(D)\nonumber\\ 
	&\overset{(\ref{eq: general Pfaffian})}{=} \epsilon_0 i^{q([D_0])} \text{Pf}(A^{K_{q+2[D_0]^*,\omega}}(G^T,x^T)). \label{eq: general Z^q}
	\end{align}
	Finally, dropping the index $\omega$ for simplicity, the Ising partition function can be written as
	\begin{eqnarray*}
		Z_{\mathcal{I}}(G,x)&=&\sum_{P\in Z_1(G;\mathbb{Z}_2)}x(P)=\sum_{\alpha\in H_1(\Sigma;\mathbb{Z}_2)}\sum_{[P]=\alpha} x(P)\\
		&\overset{\text{Lem.}\, \ref{lem: Brown}(i)}{=}&\sum_{\alpha\in H_1(\Sigma;\mathbb{Z}_2)}\bigg(\frac{1}{2^{b_1/2}}\sum_{q\in\mathcal{Q}(\Sigma,\omega)}\exp(i\pi/4)^{-\text{Br}(q)}i^{q(\alpha)}\bigg)\sum_{[P]=\alpha} x(P)\\
		&=&\frac{1}{2^{b_1/2}}\sum_{q\in \mathcal{Q}(\Sigma,\omega)}\exp(i\pi/4)^{-\text{Br}(q)}\sum_{\alpha\in H_1(\Sigma;\mathbb{Z}_2)}i^{q(\alpha)}\sum_{[P]=\alpha} x(P)\\
		&=&\frac{1}{2^{b_1/2}}\sum_{q\in \mathcal{Q}(\Sigma,\omega)}\exp(i\pi/4)^{-\text{Br}(q)}Z^q(G,x)\\
		&\overset{(\ref{eq: general Z^q})}{=}&\frac{\epsilon_0}{2^{b_1/2}}\sum_{q\in \mathcal{Q}(\Sigma,\omega)}\exp(i\pi/4)^{-\text{Br}(q)+2q([D_0])}\text{Pf}(A^{K_{q+2[D_0]^*}}(G^T,x^T))\\
		&\overset{\text{Lem.}\, \ref{lem: Brown}(ii)}{=}&\frac{\epsilon_0}{2^{b_1/2}}\sum_{q\in \mathcal{Q}(\Sigma,\omega)}\exp(i\pi/4)^{-\text{Br}(q+2[D_0]^*)}\text{Pf}(A^{K_{q+2[D_0]^*}}(G^T,x^T))\\
		&=&\frac{\epsilon_0}{2^{b_1/2}}\sum_{q\in\mathcal{Q}(\Sigma,\omega)}\exp(i\pi/4)^{-\text{Br}(q)}\text{Pf}(A^{K_{q}}(G^T,x^T)).
	\end{eqnarray*} This completes the proof of Theorem $\ref{theo: general}$.
\end{proof}
Now let us show the proof of Theorem $\ref{theo: second}$.
\begin{proof}[Proof of Theorem $\ref{theo: second}$] By the same manner of the proof of Theorem $\ref{theo: first}$, using Theorem $\ref{theo: general}$ and Lemma $\ref{lem: Brown}$ part $(ii)$ one can write 
	\begin{eqnarray}
	Z_{\mathcal{I}}(G,x)&=&\frac{1}{2^{b_1/2}}\bigg| \sum_{q\in \mathcal{Q}(\Sigma,\omega)}\exp(i\pi/4)^{\text{Br}(\tilde{q}_0)-\text{Br}(q)}\text{Pf}(A^{K_q}(G^T,x^T))\bigg| \nonumber\\
	&=&\frac{1}{2^{b_1/2}}\bigg| \sum_{\Delta\in H_1(\Sigma;\mathbb{Z}_2)}i^{\tilde{q}_0(\Delta)}\text{Pf}(A^{K_{\tilde{q}_0+2\Delta^*}}(G^T,x^T))\bigg|. \nonumber
	\end{eqnarray}
	Now recall that if $\Sigma=\Sigma_g\#\mathcal{K}$, the set $\mathcal{B}'$ given by $\{[a_1],\dots,[a_g],[b_1],\dots,[b_g],[a],[b]\}$ is a basis of $H_1(\Sigma;\mathbb{Z}_2)$. In this case $b_1=2g+2$ and each element $\Delta\in H_1(\Sigma;\mathbb{Z}_2)$ can be written as $\Delta=\sum_{i=1}^g\epsilon_i[a_i]+\sum_{i=1}^g\epsilon'_i[b_i]+\alpha[a]+\beta[b]$, and hence can be identified with $(\epsilon,\epsilon',\alpha,\beta)\in \mathbb{Z}_2^{2g+2}$. Using the fact that $[a]\cdot[b]=1$ while $[a_i]\cdot[b_j]=\delta_{ij}$, $[a_i]\cdot[a_j]=[b_i]\cdot[b_j]=0$ and $[a]\cdot[a_i]=[a]\cdot[b_j]=[b]\cdot[a_i]=[b]\cdot[b_j]=0$ for every $i,j=1,\dots,g$, by definition of $\tilde{q}_0$ and quadratic enhancements we get  
	\begin{eqnarray}
	\tilde{q}_0(\Delta)&=&\tilde{q}_0(\sum_{i=1}^g\epsilon_i[a_i]+\sum_{i=1}^g\epsilon'_i[b_i]+\alpha[a]+\beta[b]) \nonumber\\
	&=&\tilde{q}_0(\sum_{i=1}^g\epsilon_i[a_i]+\sum_{i=1}^g\epsilon'_i[b_i])+\tilde{q}_0(\alpha[a]+\beta[b])\nonumber\\
	&=&2\sum_{i=1}^g\epsilon_i\epsilon'_i +2\alpha\beta+3\alpha. \nonumber
	\end{eqnarray}
	Moreover the orientation $K_{\tilde{q}_0+2\Delta^*}$, which is obtained by inverting $K_0=K_{\tilde{q}_0}$ on every edge $e$ such that $\tilde{q}_0([e])+2\Delta^*([e])\neq \tilde{q}_0([e])$, or equivalently on every edge $e$ such that $1=\Delta^*([e])=\Delta\cdot [e]$, is exactly $K_{\epsilon,\epsilon',\alpha,\beta}$ defined before. This fact together with the equations above leads to the first formula stated in Theorem $\ref{theo: second}$. The case $\Sigma=\Sigma_g\#\mathbb{R}P^2$ is treated similarly, concluding the proof.
\end{proof}
\bibliographystyle{plain}
\bibliography{reference}
\end{document}